\documentclass{llncs}

\usepackage{amsmath,amssymb,amsfonts,braket}
\usepackage{mathtools}
\usepackage{graphicx}

\usepackage{tikz}
\usetikzlibrary{decorations.pathreplacing, positioning, backgrounds, fit, calc}

\usepackage{pgfplots}
\pgfplotsset{compat=newest}

\usepackage{quantikz}

\usepackage{subcaption}
\usepackage{booktabs}
\usepackage[ruled,vlined,linesnumbered,longend]{algorithm2e}
\usepackage{placeins}
\DontPrintSemicolon
\SetKwInput{KwIn}{Input}
\SetKwInput{KwOut}{Output}
\usepackage{textcomp}
\usepackage{xcolor}
\usepackage[utf8]{inputenc}
\usepackage[T1]{fontenc}
\usepackage[bookmarks=false]{hyperref}
\usepackage{orcidlink}
\usepackage{cleveref}
\usepackage{todonotes}

\spnewtheorem{assumption}{Assumption}{\bfseries}{\itshape}
\spnewtheorem*{informaltheorem}{Theorem (informal)}{\bfseries}{\itshape}

\newcommand{\Z}{\mathbb{Z}}
\newcommand{\R}{\mathbb{R}}
\newcommand{\N}{\mathbb{N}}
\newcommand{\Zq}{\mathbb{Z}_q}

\newcommand{\vecb}[1]{\mathbf{#1}}     
\newcommand{\mat}[1]{\mathbf{#1}}      
\newcommand{\norm}[1]{\left\|#1\right\|}
\newcommand{\inner}[2]{\langle #1,\, #2 \rangle}

\renewcommand{\ket}[1]{|#1\rangle}
\renewcommand{\bra}[1]{\langle #1|}

\renewcommand{\orcidID}[1]{\,\orcidlink{#1}}

\title{Improved Dual Attack and Trapdoor Sampling via Quantum Rejection Sampling}

\author{Cong Ling\orcidID{0000-0001-7873-4862} \and Hao Yan\orcidID{0009-0007-4464-2450} \and Nicholas Zhao\orcidID{0009-0002-0281-2435}}

\institute{Imperial College London, United Kingdom\\
\email{\{c.ling,h.yan22,n.zhao22\}@imperial.ac.uk}}

\begin{document}

\maketitle
\pagestyle{plain}
\thispagestyle{plain}

\begin{abstract}
In this work, we revisit the dual attack and GPV trapdoor sampling, focusing on the lattice Gaussian sampling term, which can be a significant bottleneck in the overall complexity. We show that this sampling step can be quantumly accelerated by combining the lower bound underlying Wang and Ling's analysis of Klein's algorithm with the quantum rejection sampling (QRS) framework proposed by Ozols et al. Specifically, this lower bound gives precisely the pointwise domination condition required for quantum rejection sampling when given coherent oracle access to a truncated Klein proposal distribution, which yields a quantum procedure for preparing the truncated dual $q$-ary lattice Gaussian with a quadratic reduction in the sampling complexity. The truncation radius is chosen so that the truncated distribution is negligibly close to the full lattice Gaussian in total variation distance. Substituting this sampler into the dual attack framework results in reduced overall attack-cost estimates. Compared with Pouly and Shen's modern dual attack under the same parameter choices, our estimates reduce the attack cost by \(9\), \(4\), and \(13\) bits for Kyber-512, Kyber-768, and Kyber-1024, respectively. We also report the corresponding estimates with modulus switching. Finally, by replacing the Markov chain Monte Carlo (MCMC) sampler with the QRS algorithm, we achieve a similar quadratic speedup in the GPV signing process.

\end{abstract}

\section{Introduction}
\label{sec:intro}

The Learning with Errors (LWE) and Short Integer Solution problems~\cite{Regev2009,MicciancioRegev2007} are foundational hardness assumptions upon which modern lattice-based cryptography rests. Their worst-case to average-case reduction---where solving average-case LWE or SIS is at least as hard as approximating certain worst-case lattice problems---established them as cornerstones of post-quantum cryptography. This is reflected in the NIST standardization of lattice-based schemes such as ML-KEM, a key-encapsulation mechanism derived from Kyber whose security relies on the hardness of Module-LWE~\cite{NISTFIPS203}, and FALCON, a digital signature scheme whose security is associated to the SIS problem~\cite{FouqueEtAlFalcon2020}. LWE also underpins primitives such as the Fully Homomorphic Encryption (FHE)~\cite{BrakerskiVaikuntanathan2011}. Understanding the concrete cost of attacks on LWE is therefore critical for assessing the security of deployed schemes.

There are many different attack methods for solving LWE, including lattice, combinatorial, algebraic, exhaustive-search, and hybrid methods. Among lattice-based approaches, the two dominant methods are primal and dual attacks; we focus on the latter. At a high level, a dual attack uses short vectors from an appropriate dual $q$-ary lattice to distinguish LWE samples from uniform; combined with a guessing step over part of the secret, this yields an attack on the decision variant of LWE. Prior dual attacks~\cite{Albrecht2017,GuoJohansson2021,MATZOV2022} establish this distinguisher via statistical heuristics in which a cosine-sum over short dual vectors concentrates around a Gaussian-derived expectation for the correct guess, and concentrates around zero for incorrect ones; the validity of this argument has recently been placed under scrutiny~\cite{DucasPulles2023}. Recent work by Pouly and Shen~\cite{PoulyShen2024} moved away from relying on such heuristics and instead rooted the attack in rigorous geometric analysis. Their attack splits the LWE secret into a guessed part and a dual part, with the matrix \(\mat{A}\) split correspondingly into \(\mat{A}_{\mathrm{guess}}\) and \(\mat{A}_{\mathrm{dual}}\). A basis of the dual lattice \(\mathcal{L}_q^\perp(\mat{A}_{\mathrm{dual}})\) is first reduced using BKZ~\cite{Schnorr1987}, and a Gaussian sampler via the Markov-chain Monte Carlo (MCMC) algorithm~\cite{wang_lattice_2019} is then used to draw many short vectors from this reduced basis. For each candidate guess, these vectors are combined into a cosine-sum score that tracks how close the corresponding residual lies to \(\mathcal{L}_q(\mat{A}_{\mathrm{dual}})\). The correct guess yields a residual equal to the LWE error modulo the lattice, while any incorrect guess shifts this residual by a nonzero coset representative whose distance to the lattice is bounded below by the shortest nonzero vector of \(\mathcal{L}_q(\mat{A})\), separating correct from incorrect guesses whenever the error is short enough relative to this minimum. Crucially, the entire argument is geometric and requires no statistical assumptions on the short vectors, which is the central contribution distinguishing this work from prior dual attacks. The quantum variant of their dual attack based on \cite{AlbrechtShen2022} accelerates the search stage of the dual attack using quantum mean estimation to approximate the cosine score associated with a candidate guess, after which quantum maximum finding is used to search over the guessed secret coordinates, replacing the classical guessing term with a quadratically reduced one. However, the dominant Gaussian sampling cost remains untouched.

The lattice Gaussian sampling step relies on MCMC methods, following a line of work developed by Wang and Ling~\cite{WangLing2015,WangLing2018GeometricErgodicity,wang_lattice_2019}, in which Klein's algorithm~\cite{Klein2000}, a randomized nearest-plane algorithm, is used as the proposal distribution and MCMC drives convergence towards the target lattice Gaussian, with convergence rate governed by the spectral gap they derive analytically. Naturally, one might first consider the quantum-walk-based methods~\cite{Szegedy2004,MagniezNayakRolandSantha2011} to accelerate the Markov chains. In particular, Wocjan et al.~\cite{WocjanAbeyesinghe2008} give an end-to-end framework for quantum sampling based on quantum walks, which can be interpreted as a quantum simulated annealing procedure. However, applying this approach directly to the Markov chain introduces additional complex structures: one must specify a cooling schedule, control the overlaps between successive target distributions, and obtain usable lower bounds on the spectral gaps of a sequence of Markov chains at different inverse temperatures, equivalently, different widths of lattice Gaussian distributions. These requirements make the resulting approach considerably more involved, and the overhead due to the requirements can offset the quadratic speedup one would hope to obtain. Our key observation is that one does not need to accelerate the chain itself. Instead, the relevant structure is already present in the Klein proposal distribution: it pointwise lower-bounds the target lattice Gaussian distribution up to a multiplicative constant. This is precisely the domination condition required for rejection sampling, whose classical form goes back to von Neumann~\cite{vonNeumann1951RandomDigits}. Under such a domination bound, the classical cost of producing one accepted target sample is proportional to the inverse of this constant; related domination-based ideas for lattice Gaussian proposals also appear in~\cite{BrakerskiLangloisPeikertRegevStehle2013}. Ozols et al.~\cite{OzolsRoettelerRoland2013} subsequently proposed a quantum analogue, quantum rejection sampling (QRS), which uses amplitude amplification to improve the dependence on this constant quadratically. This perspective has also been used in quantum state-generation problems, including the quantum Metropolis sampling algorithm~\cite{temme2011quantumMetropolis}.
This work targets the gap left open by these results: the cost of the Gaussian sampling step, which remains classical and now constitutes the dominant bottleneck of the attack. The urgency of this gap has only grown since~\cite{QuXu2025} further reduced the cost of the guessing stage, widening the disparity between the quantum-accelerated and classical components of the pipeline. We seek to close this gap by giving a quantum algorithm for the sampling step itself, thereby reducing the overall dual attack complexity.

\paragraph{\textbf{Contributions.}}
The main contribution of this work is a quantum algorithm for the discrete Gaussian sampling step of the dual attack framework in~\cite{PoulyShen2024}, replacing the classical MCMC sampler with a QRS procedure and yielding a quadratic reduction in the $\Delta$ term \eqref{eq:Delta}. First, we establish a lower bound on the truncated variant of Klein's probability distribution, a prerequisite for embedding it as the proposal oracle in QRS. We then instantiate QRS with this truncated Klein oracle to sample from a truncated lattice Gaussian whose total variation distance from the full lattice Gaussian is made negligible by an appropriate choice of truncation radius. Finally, we integrate this quantum sampler into the dual attack pipeline and tabulate the resulting end-to-end attack costs for Kyber-512, Kyber-768, and Kyber-1024, both without and with modulus switching; see Table~\ref{tab:qrs-vs-pouly-shen-no-ms} and Table~\ref{tab:qrs-vs-pouly-shen-ms}. This gives lower overall dual-attack cost estimates. We state our main result here informally:

\begin{theorem}[Informal]
\label{thm:qrs-informal}
Let \(\mathcal L=\mat B\Z^n\) and $\mathcal X_{\mathcal R}:=\{\vecb{x}\in\Z^n:\norm{\mat B\vecb{x}-\vecb c}\le \mathcal R\}$. Given access to the truncated Klein proposal oracle
\begin{equation}
    O_{Q_{\mathcal{R}}}:\ket{0}
\mapsto
\sum_{\vecb{x}\in \mathcal{X}_{\mathcal{R}}}
\sqrt{Q_{\mathcal{R}}(\vecb{x})}
\ket{\vecb{x}},
\end{equation}
its inverse $O_{Q_{\mathcal R}}^\dagger$, a rotation \(R_a\) operator. Given the Proposition~\ref{prop:klein-domination-truncated-target} gives the pointwise domination bound
\begin{equation}
Q_{\mathcal R}(\vecb{x})\ge p_{\mathcal R}\pi_{\mathcal R}(\vecb{x}),
\qquad
p_{\mathcal R}:=\frac{\Delta_{\mathcal R}}{Q(\mathcal X_{\mathcal R})}\in(0,1].
\end{equation}
Then QRS prepares the truncated lattice Gaussian state $\ket{\pi_{\mathcal R}}
=
\sum_{\vecb{x}\in\mathcal X_{\mathcal R}}
\sqrt{\pi_{\mathcal R}(\vecb{x})}\ket{\vecb{x}}$
using
\begin{equation}
\mathcal O\left(\frac{1}{\sqrt{\Delta_{\mathcal R}}}\right)
\end{equation}
oracle queries, where \(\Delta_{\mathcal R}
=
\dfrac{
\rho_{s,\vecb c}(\mat B\mathcal X_{\mathcal R})
}{
\prod_{i=1}^{n}\rho_{s_i}(\Z)
}\) is the pointwise ratio between the truncated Klein distribution \(Q_{\mathcal R}\) and the lattice Gaussian distribution \(\pi_{\mathcal R}\), as defined in Proposition~\ref{prop:klein-domination-truncated-target}. The total number of qubits is \(\left\lceil \log_2 |\mathcal X_{\mathcal R}| \right\rceil+1\).
\end{theorem}

Measurement in the computational basis outputs a sample \(\mat B\vecb{x}\) from the truncated lattice Gaussian. We then use this result to speed up the modern dual attack framework in the following theorem.
\begin{theorem}[Informal]
Under the same parameter setting as Theorem~$6$ of \cite{PoulyShen2024}, with a set of lattice points lying inside the Euclidean ball of radius $\mathcal R
\ge
qs\left(
\sqrt{\frac{m}{2\pi}}
+
\sqrt{\frac{\log(N/\varepsilon_{\mathrm{tail}})}{\pi}}
\right)$, there exists an algorithm using QRS such that the end-to-end complexity for the dual attack is
\begin{equation}
\label{eq:dual-end-to-end-complexity_1}
\operatorname{poly}(m,n)(N+q^{n_{\mathrm{guess}}})
+
T_{\mathrm{BKZ}}(m,\beta)
+
N\mathcal O\left(\frac{1}{\sqrt{\Delta_{\mathcal{R}}}}\right).
\end{equation}
\end{theorem}

Since \(\Delta_{\mathcal R}=\alpha_{\mathcal R}\Delta\), where \(\alpha_{\mathcal R}\) 
denotes the Gaussian mass captured by the ball, Corollary~\ref{cor:tail-bound} lets us choose \(\mathcal R\) so that \(\alpha_{\mathcal R}\rightarrow 1\). Here \(\Delta\)
denotes the main complexity term incurred by MCMC in \cite{PoulyShen2024} (cf. \eqref{eq:delta-PoulyShen2024}). Hence \(\Delta_{\mathcal R}\rightarrow \Delta\), and substituting this into \eqref{eq:dual-end-to-end-complexity_1} yields a lower overall dual attack cost via the quadratic improvement in the sampling complexity.

Beyond its application to attacks, we demonstrate a constructive use case for the proposed QRS algorithm: accelerating GPV trapdoor sampling. The foundational technique of the GPV signature scheme is discrete Gaussian sampling over a trapdoor lattice \cite{GenPeiVai2008}. To prevent private key leakage, its security critically depends on the output distribution of the discrete Gaussian sampling remaining completely oblivious to the specific basis used during the process. While the original GPV scheme relied on Klein's algorithm, any sufficiently secure Gaussian sampling algorithm can be adapted for GPV signatures. In particular, Wang and Ling \cite{wang_lattice_2019} proposed an MCMC sampler based on the independent MHK algorithm for GPV-style trapdoor sampling.

As shown in \cite{GenPeiVai2008}, the security of GPV signing reduces to the hardness of the Inhomogeneous Short Integer Solution (ISIS) problem with an approximation factor of $\sqrt{n/(2\pi)}\cdot s$. Consequently, the parameter $s$ is the defining characteristic of a discrete Gaussian sampler in this context. This introduces a tradeoff between security and computational efficiency when using MCMC~\cite{wang_lattice_2019} for trapdoor sampling: a smaller parameter $s$ enhances security but demands a longer running time. In particular, above smoothing, if
\(B_{\max}=\max_i\|\widehat{\vecb{b}}_i\|\) and
\(s\ge \sqrt{\gamma}B_{\max}\) for \(\gamma\ge1\), their bound~\cite[Eq. 77]{wang_lattice_2019} gives an inverse-gap at most
\[
    T_{\mathrm{MCMC}}\le   \vartheta_3(\gamma)^n(1+2\varepsilon),
\]
and hence mixing time, where \(d\) is the lattice dimension and \(\vartheta_3\) is the Jacobi theta
function.

Our QRS sampler can replace this MCMC Gaussian sampler. Under the same truncation argument as
in Theorem~\ref{thm:qrs}, the sampling factor is quadratically reduced, giving
the corresponding bound
\[
    T_{\mathrm{QRS}}
    \lesssim
    \vartheta_3(\gamma)^{n/2}\sqrt{1+2\varepsilon}.
\]
For the FALCON-512 parameters considered by
Wang and Ling, this changes the \(\gamma=2\) upper bound from
approximately \(45.49\) to \(6.74\).

\paragraph{\textbf{Related and concurrent work.}}
We became aware of concurrent and independent work by Chevignard et al.~\cite{cryptoeprint:2026/984}, who also use quantum rejection sampling for discrete Gaussian sampling over lattices. Their work gives an explicit coordinate-wise quantum implementation of a Klein-type proposal sampler over a finite box-truncated support, and then applies quantum rejection sampling to transform this proposal state into a state whose measurement distribution approximates the corresponding finite lattice Gaussian. They further develop two applications: a quantum speedup for the modern dual-attack framework and a quantum speedup for the Gaussian-sampling component of the SIS algorithm proposed by Bollauf et al.~\cite{bollauf2026solvingsisnormgaussian}. In particular, their full quantum dual-attack framework avoids QRACM by replacing the precomputed list oracle with a quantum state for a discrete Gaussian, combined with quantum mean estimation based on the previous work of~\cite{AlbrechtShen2022}. The QRACM-free variant naturally becomes feasible once one has a quantum algorithm to prepare a quantum state for a lattice Gaussian distribution, since this state can be used directly inside the quantum mean estimation subroutine instead of relying on QRACM to load a precomputed list.

Our work differs in scope and level of abstraction. We formulate the quantum rejection sampling step as an oracle and high-level circuit construction for performing lattice Gaussian sampling. Instead of giving a coordinate-by-coordinate implementation over the box-truncated region, we work with a spherical truncation. We then apply this construction to speed up the sampling subroutine of the modern dual attack framework of~\cite{PoulyShen2024}, including table comparisons between our attack-cost model and theirs under the same parameter constraints, both with and without modulus switching. Finally, we apply our quantum sampler to FALCON's trapdoor sampling, obtaining a quadratic speedup in signing, which we verify with numerical simulations.

\paragraph{\textbf{Paper organization.}}
\Cref{sec:prelim} gives the basic preliminaries on lattices, LWE, discrete Gaussians, Klein's algorithm, the MCMC time complexity, and the quantum primitives used in our quantum algorithm. In \Cref{sec:qrs}, we establish the pointwise domination bound needed for QRS and use it to obtain a quantum lattice Gaussian sampler with quadratically improved sampling complexity, up to negligible truncation error. In \Cref{sec:dual-complexity}, we recall the modern dual-attack framework of~\cite{PoulyShen2024} and incorporate our quantum sampler into the sampling component of their complexity model. We then compare the resulting dual-attack estimates with Table~1 of~\cite{PoulyShen2024}, using the same parameters, both with and without modulus switching. In \Cref{sec:falcon-trapdoor-sampling}, we apply the proposed QRS algorithm to accelerate trapdoor sampling. Finally, \Cref{sec:conclusion} concludes the paper.

\section{Preliminaries}
\label{sec:prelim}

We provide the following definitions and results we will use in our main theorems:

\subsection{Classical Preliminaries}
\label{ssec:lattice-classical-prelim}

\begin{definition}[Lattice]
\label{definition:lattice}
A lattice \(\mathcal{L}\subset \R^n\) is a discrete additive subgroup generated by the set of integer combinations of \(n\) linearly independent vectors \(\vecb{b}_1,\ldots,\vecb{b}_n\), where \(\mat{B}=[\vecb{b}_1,\ldots,\vecb{b}_n] \in \R^{n\times n}\):
\begin{equation}
\label{eq:lattice-definition}
\mathcal{L}(\mat{B})
=
\left\{
\sum_{i=1}^{n} x_i \vecb{b}_i : x_i \in \Z
\right\}.
\end{equation}
\end{definition}

\begin{definition}[Dual Lattice]
\label{def:dual-lattice}
Let \(\mathcal{L} \subseteq \R^n\) be a full-rank lattice. The dual lattice \(\mathcal{L}^*\) is defined as
\begin{equation}
\label{eq:dual-lattice-definition}
\mathcal{L}^*
=
\left\{
\vecb{y} \in \R^n :
\langle \vecb{y}, \vecb{x} \rangle \in \Z
\text{ for all }
\vecb{x} \in \mathcal{L}
\right\},
\end{equation}
where \(\langle \cdot,\cdot \rangle\) denotes the standard inner product in \(\R^n\).
\end{definition}

\begin{definition}[\(q\)-ary lattices]
\label{def:q-ary-lattice}
For a matrix \(\mat{A} \in \Zq^{m \times n}\), the \(q\)-ary lattices associated with \(\mat{A}\) are
\begin{equation}
\label{eq:q-ary-lattice}
\begin{aligned}
\mathcal{L}_q(\mat{A})
&=
\left\{
\vecb{z} \in \Z^m :
\exists \vecb{x}\in\Zq^n,\ \vecb{z} \equiv \mat{A}\vecb{x} \pmod q
\right\},\\
\mathcal{L}_q^{\perp}(\mat{A})
&=
\left\{
\vecb{z} \in \Z^m :
\mat{A}^{\top}\vecb{z} \equiv 0 \pmod q
\right\}.
\end{aligned}
\end{equation}
\end{definition}

\begin{remark}
Although \(\mathcal{L}_q^{\perp}(\mat{A})\) is standardly called the dual \(q\)-ary lattice, it is not identical to the Euclidean dual lattice \(\mathcal{L}_q(\mat{A})^*\) from Definition~\ref{def:dual-lattice}. Rather, the two are related by
\begin{equation}
\label{eq:q-ary-dual-relation}
\mathcal{L}_q(\mat{A})^*
=
\frac{1}{q}\mathcal{L}_q^{\perp}(\mat{A}).
\end{equation}
Thus \(\mathcal{L}_q^{\perp}(\mat{A})\) denotes the integer lattice of vectors orthogonal to the columns of \(\mat{A}\) modulo \(q\).
\end{remark}

\begin{definition}[Learning with Errors~\cite{Regev2009}]
\label{def:lwe}
Let \(m,n,q \in \N\), and let \(\chi\) be a distribution over \(\Zq\). The decisional \(\mathrm{LWE}_{n,q,\chi}\) problem with \(m\) samples is to distinguish
\((\mat{A},\mat{A}\vecb{s}+\vecb{e})\) from \((\mat{A},\vecb{u})\), where
\(\mat{A} \gets \Zq^{m \times n}\),
\(\vecb{s} \gets \Zq^n\),
\(\vecb{e} \gets \chi^m\), and
\(\vecb{u} \gets \Zq^m\).
\end{definition}

\begin{definition}[Discrete Gaussian Distribution]
For any \(\mathcal L \subset \R^n\), \(s > 0\), and \(\vecb{c}\in\R^n\), the discrete Gaussian distribution over \(\mathcal L\) is 
\begin{equation}
\mathcal D_{\mathcal L+\vecb{c},s}(\vecb{x}) = \frac{\rho_{s,\vecb{c}}(\vecb{x})}{\rho_{s,\vecb{c}}(\mathcal L)}
\end{equation}
for all \(\vecb{x} \in \mathcal L\), where \(\rho_{s,\vecb{c}}(\vecb{x})=\exp(-\pi\norm{\vecb{x}-\vecb{c}}^2/s^2)\) and \(\rho_{s,\vecb{c}}(\mathcal L)=\sum_{\vecb{y}\in\mathcal L}\rho_{s,\vecb{c}}(\vecb{y})\).
\end{definition}

See Figure \ref{fig:discrete-gaussian} for discrete Gaussians with different width $s$, where $s$ controls the sampling complexity. Most algorithms sample at or above the smoothing parameter \cite{MicciancioRegev2007}.

\begin{figure}[ht]
    \centering
    \begin{subfigure}[b]{0.48\textwidth}
        \centering
        \begin{tikzpicture}
            \begin{axis}[
                hide axis, 
                width=\linewidth,
                height=0.7\linewidth,
                view={60}{30},
                domain=-15:15, y domain=-15:15,
                samples=31, samples y=31,
            ]
                \addplot3[
                    only marks,
                    mark size=0.45pt,
                    scatter,
                    scatter src=z,
                ]
                ({x},{y},{exp(-pi*(x^2+y^2)/100)});
            \end{axis}
        \end{tikzpicture}
        \caption{$s=10$}
    \end{subfigure}
    \hfill
    \begin{subfigure}[b]{0.48\textwidth}
        \centering
        \begin{tikzpicture}
            \begin{axis}[
                hide axis, 
                width=\linewidth,
                height=0.7\linewidth,
                view={60}{30},
                domain=-15:15, y domain=-15:15,
                samples=31, samples y=31,
            ]
                \addplot3[
                    only marks,
                    mark size=0.45pt,
                    scatter,
                    scatter src=z,
                ]
                ({x},{y},{exp(-pi*(x^2+y^2)/16)});
            \end{axis}
        \end{tikzpicture}
        \caption{$s=4$}
    \end{subfigure}
    \caption{Discrete Gaussian mass on $\Z^2\cap[-15,15]^2$ with Gaussian parameter $s=10$ (left) and $s=4$ (right). The plotted values show the unnormalized mass underlying $\mathcal D_{\Z^2,s,\vecb{0}}(\vecb{x})=\rho_{s,\vecb{0}}(\vecb{x})/\rho_{s,\vecb{0}}(\Z^2)$.}
    \label{fig:discrete-gaussian}
\end{figure}
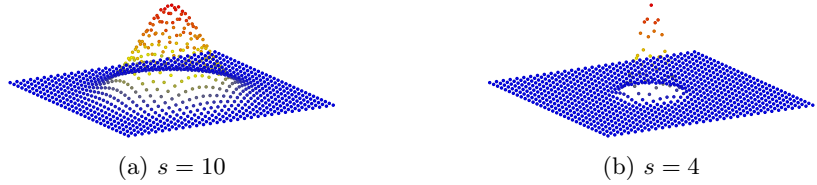

\begin{corollary}[\cite{StephensDavidowitz2017}]
\label{cor:tail-bound}
Let \(\mathcal{L}\subset\R^m\) be a lattice, let \(\vecb{c}\in\R^m\), let \(s>0\) and let \(\varepsilon_{\mathrm{tail}}\in(0,1)\). If
\begin{equation}
\mathcal{R}
\ge
s
\left(
\sqrt{\frac{m}{2\pi}}
+
\sqrt{
\frac{\log(1/\varepsilon_{\mathrm{tail}})}{\pi}
}
\right),
\end{equation}
then
\begin{equation}
\rho_s\left(
(\mathcal{L}-\vecb{c})\setminus \mathcal B_m(\mathcal{R})
\right)
\le
\varepsilon_{\mathrm{tail}}
\rho_s(\mathcal{L}).
\end{equation}
Here \(\mathcal B_m(\mathcal{R})\) denotes the radius-\(\mathcal{R}\) Euclidean ball in \(\R^m\).
\end{corollary}

\begin{definition}[Klein Probability Distribution, \cite{Klein2000}]
\label{def:klein-probability}
Let \(\mathcal L=\mat{B}\Z^n\), where \(\mat{B}=\mat{Q}\mat{R}\) is the QR decomposition, and set \(\vecb{c}'=\mat{Q}^{\top}\vecb{c}\). For \(\vecb{x}\in\Z^n\), define
\begin{equation}
s_i:=\frac{s}{|r_{ii}|},
\qquad
\mu_i(\vecb{x}_{>i}):=
\frac{c'_i-\sum_{j>i}r_{ij}x_j}{r_{ii}}.
\end{equation}
Then the Klein probability distribution is
\begin{equation}
Q(\vecb{x})
:=
\prod_{i=1}^n
\frac{\rho_{s_i,\mu_i(\vecb{x}_{>i})}(x_i)}
{\rho_{s_i,\mu_i(\vecb{x}_{>i})}(\Z)}
=
\frac{\rho_{s,\vecb{c}}(\mat{B}\vecb{x})}
{\prod_{i=1}^n\rho_{s_i,\mu_i(\vecb{x}_{>i})}(\Z)},
\end{equation}
where \(\rho_{s,\vecb{c}}(\mat{B}\vecb{x})=\prod_{i=1}^n\rho_{s_i,\mu_i(\vecb{x}_{>i})}(x_i)\).
\end{definition}

The product form in Definition~\ref{def:klein-probability} is exactly the distribution generated by Klein's Algorithm (see Algorithm~\ref{alg:klein}), in the form later used by Gentry et al.~\cite{GenPeiVai2008} with a short basis to sample from lattice Gaussian distributions. Given \(\mat{B}=\mat{Q}\mat{R}\) and
\(\vecb{c}'=\mat{Q}^{\top}\vecb{c}\), the sampler generates the coefficient vector \(\vecb{x}\in\mathbb Z^n\) sequentially in reverse order from \(x_n\) down to \(x_1\). At each step \(i\), after the higher-index coordinates \(x_{i+1},\ldots,x_n\) have been fixed, it samples $x_i \leftarrow \mathcal D_{\mathbb Z,s_i,\mu_i(\vecb{x}_{>i})}.$
Multiplying the conditional probabilities over \(i=n,\ldots,1\) gives precisely \(Q(\vecb{x})\). Thus Algorithm~\ref{alg:klein} samples $\vecb{x}$ from $Q$ in the coefficient space and outputs the corresponding lattice vector \(\mat{B}\vecb{x}\).

\begin{algorithm}
\caption{Klein's Algorithm~\cite{Klein2000}}
\label{alg:klein}
\KwIn{\(\mat{B}, s, \vecb{c}\)}
\KwOut{\(\mat{B}\vecb{x} \in \mathcal L\)}
Let \(\mat{B} = \mat{Q}\mat{R}\) and \(\vecb{c}' = \mat{Q}^{\top}\vecb{c}\)\;
\For{\(i = n,\ldots,1\)}{
    Let \(s_i = \frac{s}{|r_{ii}|}\) and
    \(\mu_i(\vecb{x}_{>i}) = \frac{c'_i-\sum_{j>i}r_{ij}x_j}{r_{ii}}\)\;
    Sample \(x_i\) from \(\mathcal{D}_{\Z,s_i,\mu_i(\vecb{x}_{>i})}\)\;
}
\KwRet{\(\mat{B}\vecb{x}\)}
\end{algorithm}

\begin{definition}[Total Variation Distance]
Let \(\mu\) and \(\nu\) be two probability distributions on a countable state space \(\Omega\). The total variation distance between \(\mu\) and \(\nu\) is
\begin{equation}
	d_{\mathrm{TV}}(\mu,\nu) := \frac{1}{2} \sum_{x \in \Omega} |\mu(x) - \nu(x)|.
\end{equation}
\end{definition}

\begin{theorem}[{\cite[Theorem 1]{wang_lattice_2019}}]
\label{thm:imhk}
Let \(\mathcal{L}\subset \R^n\) be a lattice with basis \(\mat{B}\), let \(s>0\), and let \(\varepsilon>0\). There exists an algorithm that outputs a sample according to a distribution \(\mathcal D^{\varepsilon}_{\mathcal{L},s}\) satisfying \(d_{\mathrm{TV}}\left(\mathcal D^{\varepsilon}_{\mathcal{L},s},\mathcal D_{\mathcal{L},s}\right)\le \varepsilon\).
The algorithm runs in time
\begin{equation}
\ln\left(\frac{1}{\varepsilon}\right)\cdot \frac{1}{\Delta}\cdot \operatorname{poly}(n),
\end{equation}
where
\begin{equation}
\frac{1}{\Delta}
=
\frac{
\prod_{i=1}^{n}\rho_{s/\norm{\widehat{\vecb{b}}_i}}(\Z)
}{
\rho_s(\mathcal{L})
}, \label{eq:Delta}
\end{equation}
and \(\widehat{\vecb{b}}_1,\ldots,\widehat{\vecb{b}}_n\) are the Gram--Schmidt vectors of \(\mat{B}\).
\end{theorem}

In \cite{wang_lattice_2019}, the MCMC algorithm performs lattice Gaussian sampling by embedding Klein's algorithm into a Metropolis--Hastings (MH) framework, where \(\mathcal D_{\mathcal L+\vecb{c},s}\) is set as the target distribution. Specifically, a candidate state is generated by Klein's algorithm, followed by the MH acceptance condition to either accept or reject it.

\subsection{Quantum Preliminaries}
\label{ssec:quantum-preliminaries}

For the quantum part, we require the following results:
\begin{definition}[Quantum Sample, \cite{AharonovTaShma2003}]
\label{def:qsample}
For a probability distribution $\nu$ on a finite or countable state space $\Omega$, its coherent quantum encoding is
\begin{equation}
	\ket{\nu}=\sum_{x\in\Omega}\sqrt{\nu(x)}\ket{x}.
\end{equation}
\end{definition}

Measuring \(|\nu\rangle\) in the computational basis returns \(x\in\Omega\) with probability \(\nu(x)\), so \(|\nu\rangle\) is the quantum analogue of a classical sample from \(\nu\).

\begin{theorem}[Amplitude amplification~\cite{Brassard2002}]
\label{thm:amplitude-amplification}
Let \(\mathcal A\) be a unitary such that
\begin{equation}
\label{eq:aa-decomposition}
\mathcal A\ket{0}
=
\sqrt{p}\ket{\psi}
+
\sqrt{1-p}\ket{\psi^\perp},
\end{equation}
where \(p\in(0,1]\), \(\ket{\psi}\in\operatorname{im}(\Pi)\) and \(\ket{\psi^\perp}\in\operatorname{im}(\mathbb I-\Pi)\) where $\Pi$ is the target subspace. Let \(S_0 := \mathbb I - 2\ket{0}\bra{0}\) be the reflection across input state $\ket{0}$, and \(S_{\Pi} := \mathbb I - 2\Pi\) be the reflection across the target subspace. Then amplitude amplification is defined by the following operations:
\begin{equation}
\label{eq:aa-iterate}
\mathcal U := -\mathcal A\, S_0\, \mathcal A^{\dagger}\, S_{\Pi}
\end{equation}
\end{theorem}

Ozols et al.~\cite{OzolsRoettelerRoland2013} introduced quantum rejection sampling as a quantum analogue of classical rejection sampling. Quantum rejection sampling uses amplitude amplification to boost the accepted branch during the procedure. Given a unitary oracle that prepares \(O_{\nu}:|0\rangle\mapsto|\nu\rangle\), if the target distribution \(\pi\) satisfies \(\nu(x)\ge p\pi(x)\) for all \(x\in\Omega\), where \(p>0\), and a coherent implementation of the acceptance probability \(p\pi(x)/\nu(x)\) is available, then QRS prepares the target quantum sample \(|\pi\rangle\) using quadratically fewer repetitions than classical rejection sampling. We use the following specialization.
\begin{lemma}[Lemma 4.1, \cite{OzolsRoettelerRoland2013}, restated]
\label{lemma:qrs}
Let \(\Omega\) be a finite state space and let \(\nu\) and \(\pi\) be probability distributions on \(\Omega\). Suppose that for some \(p\in(0,1]\), \(\nu(x)\ge p\pi(x)\) for every \(x\in\Omega\). Assume access to a unitary preparation oracle \(O_{\nu}:|0\rangle\mapsto|\nu\rangle=\sum_{x\in\Omega}\sqrt{\nu(x)}|x\rangle\), and to a controlled rotation implementing the acceptance probability
\begin{equation}
\label{eq:acceptance}
a(x):=p\frac{\pi(x)}{\nu(x)}.
\end{equation}
Then quantum rejection sampling prepares
\begin{equation}
|\pi\rangle=
\sum_{x\in\Omega}\sqrt{\pi(x)}|x\rangle
\end{equation}
using \(\mathcal O(1/\sqrt p)\) queries to \(O_{\nu}\), \(O_{\nu}^{\dagger}\), and the controlled-rotation oracle.
\end{lemma}

\begin{remark}
    A caveat is that standard amplitude amplification requires knowledge of the success probability \(p\) or the rotation angle, in order to choose the number of reflections without overshooting. If \(p\) is not known exactly but a certified lower bound \(p_{\min}\leq p\) is available, one can instead use fixed-point amplitude amplification~\cite{YoderLowChuang2014}, which uses
\[
\mathcal O\left(\frac{\log(1/\varepsilon)}{\sqrt{p_{\min}}}\right)
\]
iterations to prepare the target state up to error \(\varepsilon\). Thus the quadratic improvement over classical rejection sampling is retained with respect to the known lower bound \(p_{\min}\). This construction avoids overshooting by using phase sequences derived from Chebyshev polynomials and results in monotonic convergence toward the target subspace.
\end{remark}

\section{Quantum Algorithm for Lattice Gaussian Sampling}
\label{sec:qrs}

We now derive the lower bound required for preparing a lattice Gaussian from the Klein probability distribution. We use truncated versions of both distributions, since embedding them into a meaningful quantum state requires finite support.

\begin{proposition}
\label{prop:klein-domination-truncated-target}
Let \(\mathcal L=\mat{B}\Z^n\) and let \(\mathcal{X}_{\mathcal{R}}:=\{\vecb{x}\in\Z^n:\norm{\mat{B}\vecb{x}-\vecb{c}}\le \mathcal{R}\}\) be a coefficient-space truncation such that
\[
\pi_{\mathcal{R}}(\vecb{x})=\frac{\rho_{s,\vecb{c}}(\mat{B}\vecb{x})}{\rho_{s,\vecb{c}}(\mat{B}\mathcal{X}_{\mathcal{R}})},
\qquad
Q_{\mathcal{R}}(\vecb{x})=\frac{Q(\vecb{x})}{Q(\mathcal{X}_{\mathcal{R}})}
\]
where \(Q(\mathcal{X}_{\mathcal{R}}), \rho_{s,\vecb{c}}(\mat{B}\mathcal{X}_{\mathcal{R}})>0\), with \(Q(\mathcal{X}_{\mathcal{R}}):=\sum_{\vecb{x}\in\mathcal{X}_{\mathcal{R}}}Q(\vecb{x})\) and \(\rho_{s,\vecb{c}}(\mat{B}\mathcal{X}_{\mathcal{R}}):=\sum_{\vecb{x}\in\mathcal{X}_{\mathcal{R}}}\rho_{s,\vecb{c}}(\mat{B}\vecb{x})\). Let \(\pi_{\mathcal{R}}\) denote the truncated lattice Gaussian on \(\mathcal{X}_{\mathcal{R}}\), and let \(Q_{\mathcal{R}}\) denote the normalized restriction of the Klein probability distribution to \(\mathcal{X}_{\mathcal{R}}\); both distributions are zero outside \(\mathcal{X}_{\mathcal{R}}\). Then for every \(\vecb{x}\in\Z^n\),
\begin{equation}
\label{eq:klein-domination-truncated-target}
Q_{\mathcal R}(\vecb{x})
\geq
p_{\mathcal R}\pi_{\mathcal R}(\vecb{x}),
\qquad
p_{\mathcal R}:=
\dfrac{\Delta_{\mathcal R}}{Q(\mathcal X_{\mathcal R})}\in(0,1],
\end{equation}
where
\begin{equation}
\Delta_{\mathcal R}
=
\dfrac{\rho_{s,\vecb{c}}(\mat{B}\mathcal X_{\mathcal R})}
{\prod_{i=1}^{n}\rho_{s_i}(\Z)}.
\end{equation}
\end{proposition}

\begin{proof}
Let \(\vecb{x}\in\Z^n\). If \(\vecb{x}\notin \mathcal X_{\mathcal R}\), then \(\pi_{\mathcal R}(\vecb{x})=Q_{\mathcal R}(\vecb{x})=0\), proving the claim. It remains to consider \(\vecb{x}\in\mathcal X_{\mathcal R}\). Using the definitions of \(Q_{\mathcal R}\), \(Q\), and \(\pi_{\mathcal R}\), we have
\begin{equation}
\begin{aligned}
\dfrac{Q_{\mathcal R}(\vecb{x})}{\pi_{\mathcal R}(\vecb{x})}
&=
\dfrac{1}{Q(\mathcal X_{\mathcal R})}
\dfrac{\rho_{s,\vecb c}(\mat B\vecb{x})}
{\prod_{i=1}^{n}\rho_{s_i,\mu_i(\vecb{x}_{>i})}(\Z)}
\dfrac{\rho_{s,\vecb c}(\mat B\mathcal X_{\mathcal R})}
{\rho_{s,\vecb c}(\mat B\vecb{x})} \\
&=
\dfrac{1}{Q(\mathcal X_{\mathcal R})}
\dfrac{\rho_{s,\vecb c}(\mat B\mathcal X_{\mathcal R})}
{\prod_{i=1}^{n}\rho_{s_i,\mu_i(\vecb{x}_{>i})}(\Z)}.
\end{aligned}
\end{equation}
By the inequality \(\rho_{s_i,\mu_i(\vecb{x}_{>i})}(\Z)\leq \rho_{s_i}(\Z)\), we get
\begin{equation}
\dfrac{Q_{\mathcal R}(\vecb{x})}{\pi_{\mathcal R}(\vecb{x})}
\geq
\dfrac{1}{Q(\mathcal X_{\mathcal R})}
\dfrac{\rho_{s,\vecb c}(\mat B\mathcal X_{\mathcal R})}
{\prod_{i=1}^{n}\rho_{s_i}(\Z)}
=
\dfrac{\Delta_{\mathcal R}}{Q(\mathcal X_{\mathcal R})}
=
p_{\mathcal R}.
\end{equation}
Hence
\begin{equation}
Q_{\mathcal R}(\vecb{x})
\geq
p_{\mathcal R}\pi_{\mathcal R}(\vecb{x}).
\end{equation}
It remains to show that \(p_{\mathcal R}\in(0,1]\). By the same bound, summing over \(\mathcal X_{\mathcal R}\) gives
\begin{equation}
Q(\mathcal X_{\mathcal R})
\geq
\dfrac{\rho_{s,\vecb c}(\mat B\mathcal X_{\mathcal R})}
{\prod_{i=1}^{n}\rho_{s_i}(\Z)}
=
\Delta_{\mathcal R}.
\end{equation}
Therefore, since \(Q(\mathcal X_{\mathcal R})>0\) and \(\Delta_{\mathcal R}>0\),
\begin{equation}
0<p_{\mathcal R}
=
\dfrac{\Delta_{\mathcal R}}{Q(\mathcal X_{\mathcal R})}
\leq 1.
\end{equation} \qed
\end{proof}

We now instantiate Lemma~\ref{lemma:qrs} with a truncated Klein oracle to prepare the quantum state that encodes the truncated lattice Gaussian as the target distribution.
\begin{definition}[Truncated Klein Oracle]
\label{def:klein-oracle}
Let \(O_{Q_{\mathcal{R}}}\) be a unitary oracle on state space $\mathcal{X}_{\mathcal{R}}$ satisfying
\begin{equation}
\label{eq:klein-oracle}
O_{Q_{\mathcal{R}}}:\ket{0}
\mapsto
\sum_{\vecb{x}\in \mathcal{X}_{\mathcal{R}}}
\sqrt{Q_{\mathcal{R}}(\vecb{x})}
\ket{\vecb{x}}.
\end{equation}
\end{definition}

\begin{theorem}[QRS preparation from the truncated Klein proposal]
\label{thm:qrs}
Let \(\mathcal L=\mat{B}\Z^n\), let \(\mathcal{X}_{\mathcal{R}}:=\{\vecb{x}\in\Z^n:\norm{\mat{B}\vecb{x}-\vecb{c}}\le \mathcal{R}\}\), and let \(Q_{\mathcal{R}}\) and \(\pi_{\mathcal{R}}\) be normalized distributions on \(\mathcal{X}_{\mathcal{R}}\) as in Proposition~\ref{prop:klein-domination-truncated-target}. Let \(\Delta_{\mathcal{R}}\) be as in Proposition~\ref{prop:klein-domination-truncated-target} and assume oracle access \eqref{eq:klein-oracle} and controlled rotation gate \(R_a\) that implements the acceptance probability \(a(\vecb{x})\). Then Algorithm \ref{alg:qrs-klein-sampler} prepares the exact state $\ket{\pi_{\mathcal{R}}}$ using
\begin{equation}
\mathcal O\left(\frac{1}{\sqrt{p_{\mathcal R}}}\right)
\leq
\mathcal O\left(\frac{1}{\sqrt{\Delta_{\mathcal R}}}\right).
\end{equation}
queries to \(O_{Q_{\mathcal{R}}}\), \(O_{Q_{\mathcal{R}}}^{\dagger}\), controlled rotations \(R_a\), $R_a^{\dag}$ and the reflections \(S_0,S_\Pi\), and uses $\left\lceil\log_2 |\mathcal X_{\mathcal R}|\right\rceil+1$ number of qubits. 

\end{theorem}

\begin{proof}
Begin with the oracle
\begin{equation}
O_{Q_{\mathcal{R}}}\ket{0}
=
\ket{Q_{\mathcal{R}}}
=
\sum_{\vecb{x}\in \mathcal{X}_{\mathcal{R}}}\sqrt{Q_{\mathcal{R}}(\vecb{x})}\ket{\vecb{x}}.
\end{equation}
Using Lemma~\ref{lemma:qrs}, for \(\vecb{x}\in \mathcal{X}_{\mathcal{R}}\), the acceptance probability can be obtained using \eqref{eq:acceptance}:

\begin{equation}
\begin{aligned}
a(\vecb{x})
&=
p_{\mathcal R}
\frac{\pi_{\mathcal{R}}(\vecb{x})}{Q_{\mathcal{R}}(\vecb{x})}
\\
&=
\frac{\Delta_{\mathcal{R}}}{Q(\mathcal{X}_{\mathcal{R}})}
\left(
\frac{\rho_{s,\vecb{c}}(\mat{B}\vecb{x})}
{\rho_{s,\vecb{c}}(\mat{B}\mathcal{X}_{\mathcal{R}})}
\right)
\left(
\frac{
Q(\mathcal{X}_{\mathcal{R}})\prod_{i=1}^n \rho_{s_i,\mu_i(\vecb{x}_{>i})}(\Z)
}{
\rho_{s,\vecb{c}}(\mat{B}\vecb{x})
}
\right)
\\
&=
\Delta_{\mathcal{R}}
\frac{
\prod_{i=1}^n \rho_{s_i,\mu_i(\vecb{x}_{>i})}(\Z)
}{
\rho_{s,\vecb{c}}(\mat{B}\mathcal{X}_{\mathcal{R}})
}
\\
&=
\prod_{i=1}^n
\frac{
\rho_{s_i,\mu_i(\vecb{x}_{>i})}(\Z)
}{
\rho_{s_i}(\Z)
},
\end{aligned}
\end{equation}
since \(\rho_{s_i,\mu_i(\vecb{x}_{>i})}(\Z)
\leq
\rho_{s_i}(\Z)\), each factor above lies in \([0,1]\). Hence \(0\leq a(\vecb{x})\leq 1\), so the corresponding controlled acceptance rotation is well-defined. Now the rotation is defined as
\begin{equation}
\label{eq:qrs-controlled-rotation}
R_a:\ket{\vecb{x}}\ket{0}
\mapsto
\ket{\vecb{x}}
\left(
\sqrt{1-a(\vecb{x})}\ket{0}
+
\sqrt{a(\vecb{x})}\ket{1}
\right),
\end{equation}
where, on the subspace spanned by \(\{\ket{\vecb{x}}:\vecb{x}\in \mathcal{X}_{\mathcal{R}}\}\), \(R_a\) is block-diagonal over \(\vecb{x}\) with
\begin{equation}
\label{eq:qrs-rotation-block}
R_a
=
\sum_{\vecb{x}\in \mathcal{X}_{\mathcal{R}}}
\ket{\vecb{x}}\bra{\vecb{x}}
\otimes
\begin{pmatrix}
\sqrt{1-a(\vecb{x})}
&
-\sqrt{a(\vecb{x})}
\\
\sqrt{a(\vecb{x})}
&
\sqrt{1-a(\vecb{x})}
\end{pmatrix}.
\end{equation}
Applying \(R_a\) to the proposal state gives
\begin{equation}
\label{eq:qrs-after-rotation}
\begin{aligned}
R_a\left(\ket{Q_{\mathcal{R}}}\ket{0}\right)
&=
\sum_{\vecb{x}\in \mathcal{X}_{\mathcal{R}}}
\sqrt{Q_{\mathcal{R}}(\vecb{x})}
\sqrt{1-a(\vecb{x})}
\ket{\vecb{x}}\ket{0}
+
\sum_{\vecb{x}\in \mathcal{X}_{\mathcal{R}}}
\sqrt{Q_{\mathcal{R}}(\vecb{x})}
\sqrt{a(\vecb{x})}
\ket{\vecb{x}}\ket{1}.
\end{aligned}
\end{equation}
The second term is the accepted branch:
\begin{equation}
\label{eq:qrs-accepted-branch}
\begin{aligned}
\sum_{\vecb{x}\in \mathcal{X}_{\mathcal{R}}}
\sqrt{Q_{\mathcal{R}}(\vecb{x})}
\sqrt{a(\vecb{x})}
\ket{\vecb{x}}\ket{1}
&=
\sum_{\vecb{x}\in \mathcal{X}_{\mathcal{R}}}
\sqrt{Q_{\mathcal{R}}(\vecb{x})}
\sqrt{
\frac{\Delta_{\mathcal R}}{Q(\mathcal X_{\mathcal R})}
\frac{\pi_{\mathcal R}(\vecb{x})}{Q_{\mathcal R}(\vecb{x})}
}
\ket{\vecb{x}}\ket{1}
\\
&=
\sqrt{\frac{\Delta_{\mathcal R}}{Q(\mathcal X_{\mathcal R})}}
\sum_{\vecb{x}\in \mathcal{X}_{\mathcal{R}}}
\sqrt{\pi_{\mathcal{R}}(\vecb{x})}
\ket{\vecb{x}}\ket{1} \\
&=
\sqrt{\frac{\Delta_{\mathcal R}}{Q(\mathcal X_{\mathcal R})}}
\ket{\pi_{\mathcal{R}}}\ket{1}.
\end{aligned}
\end{equation}
The normalized accepted state is \(\ket{\pi_{\mathcal{R}}}\ket{1}\), so discarding the accept qubit gives \(\ket{\pi_{\mathcal{R}}}\). Finally, taking the inverse gives \(\dfrac{1}{\sqrt{p_{\mathcal R}}}\leq\sqrt{\dfrac{Q(\mathcal X_{\mathcal R})}{\Delta_{\mathcal R}}}\leq\dfrac{1}{\sqrt{\Delta_{\mathcal R}}}\), since \(Q(\mathcal{X}_{\mathcal R})\leq 1\). Therefore the query complexity is
\begin{equation}
\label{eq:qrs-final-query-bound}
\mathcal O\left(
\frac{1}{\sqrt{\Delta_{\mathcal{R}}}}
\right).
\end{equation}
Given our oracle model, we need not consider the registers required to implement \(O_{Q_{\mathcal R}}\), \(R_a\), or their inverses. The QRS circuit only uses a coefficient register and one accept coin; hence the number of qubits is
\begin{equation}
\label{eq:qrs-visible-register-count}
N_{\mathrm{qubits}}=\left\lceil\log_2 |\mathcal X_{\mathcal R}|\right\rceil+1.
\end{equation}
\qed
\end{proof}

\begin{algorithm}[tb]
\caption{Quantum Rejection Sampler for Lattice Gaussians}
\label{alg:qrs-klein-sampler}
\KwIn{Basis \(\mat{B}\) of \(\mathcal L=\mat{B}\Z^n\), width \(s>0\), center \(\vecb{c}\), radius \(\mathcal R\), \(\mathcal{X}_{\mathcal R}\gets\{\vecb{x}\in\Z^n:\|\mat{B}\vecb{x}-\vecb{c}\|\le\mathcal R\}\), oracle \(O_{Q_{\mathcal R}}\), acceptance rotation \(R_a\), $p_{\mathcal R}$, register $\mathcal X$ and coin register initialized in $\ket{0}_c$}
\KwOut{A lattice vector \(\mat{B}\vecb{x}\).}
\(p_{\mathcal R}\gets\Delta_{\mathcal R}/Q(\mathcal{X}_{\mathcal R})\)\;
Prepare \(\ket{Q_{\mathcal R}}\) on register \(\mathcal X\) via \(O_{Q_{\mathcal R}}\)\;
Apply \(R_a\) on coin register to obtain the initial state \(\mathcal{A}_{\mathcal R}\ket{0}_{\mathcal X}\ket{0}_c\)\;
\For{$\mathcal O(1/\sqrt{p_{\mathcal R}})$ \textnormal{times}}{
  Apply $\mathcal U_{\mathcal R} = \mathcal A_{\mathcal R} S_0 \mathcal A_{\mathcal R}^{\dagger} S_{\Pi}$\;
}
Measure the coin register\;
\KwRet{$\mat{B}\vecb{x}$}
\end{algorithm}

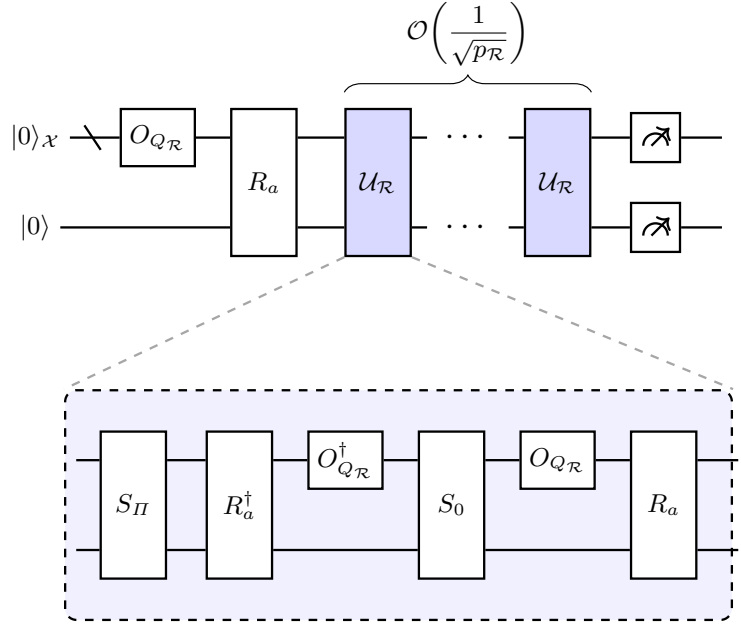
\begin{figure}[tb]
\centering
\resizebox{0.8\textwidth}{!}{%
\begin{tikzpicture}[
  wire/.style={thick},
  gate/.style={draw, thick, fill=white, minimum height=0.7cm, minimum width=0.8cm, font=\small, inner sep=3pt},
  bluegate/.style={gate, fill=blue!15},
  meter/.style={draw, thick, fill=white, minimum height=0.6cm, minimum width=0.6cm}
]

\node[font=\small] (lblX) at (0,0) {$\ket{0}_{\mathcal{X}}$};
\node[font=\small] (lblC) at (0,-1.1) {$\ket{0}$};

\node[gate] (Oqr) at (1.5, 0) {$O_{Q_{\mathcal{R}}}$};
\node[gate, minimum height=1.8cm] (Ra) at (2.8, -0.55) {$R_a$};
\node[bluegate, minimum height=1.8cm] (U1) at (4.2, -0.55) {$\mathcal{U}_{\mathcal{R}}$};

\node[font=\large] (dotsT) at (5.3, 0) {$\cdots$};
\node[font=\large] (dotsC) at (5.3, -1.1) {$\cdots$};

\node[bluegate, minimum height=1.8cm] (U2) at (6.4, -0.55) {$\mathcal{U}_{\mathcal{R}}$};
\node[meter] (metT) at (7.6, 0) {};
\node[meter] (metC) at (7.6, -1.1) {};

\draw[wire] (lblX) -- (Oqr.west);
\draw[wire] (Oqr.east) -- ([yshift=0.55cm]Ra.west);
\draw[wire] (lblC) -- ([yshift=-0.55cm]Ra.west);

\draw[wire] ([yshift=0.55cm]Ra.east) -- ([yshift=0.55cm]U1.west);
\draw[wire] ([yshift=-0.55cm]Ra.east) -- ([yshift=-0.55cm]U1.west);

\draw[wire] ([yshift=0.55cm]U1.east) -- (4.8, 0);
\draw[wire] ([yshift=-0.55cm]U1.east) -- (4.8, -1.1);

\draw[wire] (5.8, 0) -- ([yshift=0.55cm]U2.west);
\draw[wire] (5.8, -1.1) -- ([yshift=-0.55cm]U2.west);

\draw[wire] ([yshift=0.55cm]U2.east) -- (metT.west);
\draw[wire] ([yshift=-0.55cm]U2.east) -- (metC.west);

\draw[wire] (metT.east) -- ++(0.5,0);
\draw[wire] (metC.east) -- ++(0.5,0);

\draw[thick] (0.6, 0.15) -- (0.8, -0.15);

\draw[thick] (metT.center) ++(-0.15,-0.1) arc (180:0:0.15);
\draw[thick,->] (metT.center) ++(-0.05,-0.1) -- ++(0.22,0.22);
\draw[thick] (metC.center) ++(-0.15,-0.1) arc (180:0:0.15);
\draw[thick,->] (metC.center) ++(-0.05,-0.1) -- ++(0.22,0.22);

\draw[decorate, decoration={brace, amplitude=6pt, raise=4pt}]
  ([yshift=0.05cm]U1.north west) -- ([yshift=0.05cm]U2.north east)
  node[midway, above=10pt, font=\small] {$\mathcal O\!\left(\dfrac{1}{\sqrt{p_{\mathcal{R}}}}\right)$};

\begin{scope}[shift={(4.2, -4.5)}]
  \node[gate, minimum height=1.8cm] (bSPi) at (-3.0, 0) {$S_{\Pi}$};
  \node[gate, minimum height=1.8cm] (bRad) at (-1.7, 0) {$R_a^{\dagger}$};
  \node[gate] (bOqrd) at (-0.4, 0.55) {$O_{Q_{\mathcal{R}}}^{\dagger}$};
  \node[gate, minimum height=1.8cm] (bS0) at (0.9, 0) {$S_{0}$};
  \node[gate] (bOqr) at (2.2, 0.55) {$O_{Q_{\mathcal{R}}}$};
  \node[gate, minimum height=1.8cm] (bRa) at (3.5, 0) {$R_a$};

  \draw[wire] (-3.7, 0.55) -- ([yshift=0.55cm]bSPi.west);
  \draw[wire] (-3.7, -0.55) -- ([yshift=-0.55cm]bSPi.west);

  \draw[wire] ([yshift=0.55cm]bSPi.east) -- ([yshift=0.55cm]bRad.west);
  \draw[wire] ([yshift=-0.55cm]bSPi.east) -- ([yshift=-0.55cm]bRad.west);

  \draw[wire] ([yshift=0.55cm]bRad.east) -- (bOqrd.west);
  \draw[wire] ([yshift=-0.55cm]bRad.east) -- ([yshift=-0.55cm]bS0.west);
  \draw[wire] (bOqrd.east) -- ([yshift=0.55cm]bS0.west);

  \draw[wire] ([yshift=0.55cm]bS0.east) -- (bOqr.west);
  \draw[wire] ([yshift=-0.55cm]bS0.east) -- ([yshift=-0.55cm]bRa.west);
  \draw[wire] (bOqr.east) -- ([yshift=0.55cm]bRa.west);

  \draw[wire] ([yshift=0.55cm]bRa.east) -- ++(0.5, 0);
  \draw[wire] ([yshift=-0.55cm]bRa.east) -- ++(0.5, 0);
\end{scope}

\begin{scope}[on background layer]
  \node[draw, dashed, thick, rounded corners, fill=blue!6, 
        inner xsep=12pt, inner ysep=14pt, 
        fit=(bSPi) (bRa) (bOqrd)] (box) {};
\end{scope}

\draw[dashed, gray!70, thick] (U1.south west) -- (box.north west);
\draw[dashed, gray!70, thick] (U1.south east) -- (box.north east);

\end{tikzpicture}%
}
\caption{Oracle-level decomposition of our quantum sampler. The top wire is the coefficient register \(X\), and the bottom wire is the coin register. The first two gates \(\mathcal A_{\mathcal R}=R_a(O_{Q_{\mathcal R}}\otimes\mathbb I_0)\) prepare the state for QRS. The operator \(\mathcal U_{\mathcal R}\) is one iterate of the amplitude amplification subroutine, whose dashed box gives its full decomposition. This is repeated \(\mathcal O\left(\frac{1}{\sqrt{p_{\mathcal R}}}\right)\) times in order to converge to the target state \(\ket{\pi_{\mathcal R}}\). Measuring coin register $c$ outputs a lattice vector \(\mat{B}\vecb{x}\) from the truncated lattice Gaussian distribution.\label{fig:qrs-klein}}
\end{figure}
Let \(\pi(\vecb{x}):=\mathcal D_{\mathcal L+\vecb c,s}(\mat B\vecb{x})\). We define the retained Gaussian mass as
\begin{equation}
\label{eq:alpha-r-definition}
\alpha_{\mathcal R}
:=
\sum_{\vecb{x}\in\mathcal X_{\mathcal R}}\pi(\vecb{x})
=
\frac{\rho_{s,\vecb c}(\mat B\mathcal X_{\mathcal R})}
{\rho_{s,\vecb c}(\mathcal L)}.
\end{equation}
For \(\vecb{x}\in\mathcal X_{\mathcal R}\), the truncated and full distributions are related by
\begin{equation}
\label{eq:pi-r-pi-alpha}
\pi_{\mathcal R}(\vecb{x})
=
\frac{\pi(\vecb{x})}{\alpha_{\mathcal R}}.
\end{equation}
Therefore,
\begin{equation}
\label{eq:pi-r-tvd}
\begin{aligned}
d_{\mathrm{TV}}(\pi,\pi_{\mathcal R})
&=
\frac{1}{2}
\sum_{\vecb{x}\in\mathcal X_{\mathcal R}}
\left|\pi(\vecb{x})-\pi_{\mathcal R}(\vecb{x})\right|
+
\frac{1}{2}
\sum_{\vecb{x}\notin\mathcal X_{\mathcal R}}
\pi(\vecb{x})\\
&=
\frac{1}{2}
\sum_{\vecb{x}\in\mathcal X_{\mathcal R}}
\pi(\vecb{x})
\left(\frac{1}{\alpha_{\mathcal R}}-1\right)
+
\frac{1}{2}
\sum_{\vecb{x}\notin\mathcal X_{\mathcal R}}
\pi(\vecb{x})\\
&=
\frac{1}{2}
\left(\frac{1}{\alpha_{\mathcal R}}-1\right)
\alpha_{\mathcal R}
+
\frac{1}{2}
(1-\alpha_{\mathcal R})\\
&=
1-\alpha_{\mathcal R}.
\end{aligned}
\end{equation}

Thus, choosing \(\mathcal R\) according to Corollary~\ref{cor:tail-bound} makes \(\alpha_{\mathcal R}\rightarrow 1\), so \(1-\alpha_{\mathcal R}\) becomes negligible and the truncated sampler is statistically close to the full lattice Gaussian. Hence, measuring this quantum state outputs \(\mat{B}\vecb{x}\) with \(\vecb{x}\leftarrow\pi_{\mathcal R}\). Additionally, if \(\alpha_{\mathcal R}\to1\), then \(1/\sqrt{\Delta_{\mathcal{R}}}\) approaches \(1/\sqrt{\Delta}\), giving the quadratic reduction in the \(\Delta\) term that dominates the MCMC sampling cost. Algorithm~\ref{alg:qrs-klein-sampler} contains details of our quantum sampler, and Figure~\ref{fig:qrs-klein} shows one instance of the QRS circuit.


\section{Quantum Speedup of Modern Dual Attack}
\label{sec:dual-complexity}

We first recall the modern dual attack of~\cite{PoulyShen2024}. Let \(q\) be a prime power and let \(\mat{A}\in\Zq^{m\times n}\), \(\vecb{s}\in\Zq^n\), \(\vecb{e}\in\Z^m\), an LWE instance is
\begin{equation}
\label{eq:lwe-instance}
\vecb{b}=\mat{A}\vecb{s}+\vecb{e} \pmod{q}.
\end{equation}
Fix a split \(n=n_{\mathrm{guess}}+n_{\mathrm{dual}}\). After permuting the columns of \(\mat{A}\) so that the first \(n_{\mathrm{guess}}\) columns correspond to the coordinates to be guessed, write
\begin{equation}
\label{eq:dual-attack-split}
\mat{A}=
\begin{pmatrix}
\mat{A}_{\mathrm{guess}} & \mat{A}_{\mathrm{dual}}
\end{pmatrix},
\qquad
\vecb{s}=
\begin{pmatrix}
\vecb{s}_{\mathrm{guess}}\\
\vecb{s}_{\mathrm{dual}}
\end{pmatrix},
\end{equation}
where \(\mat{A}_{\mathrm{guess}}\in\Zq^{m\times n_{\mathrm{guess}}}\),
\(\mat{A}_{\mathrm{dual}}\in\Zq^{m\times n_{\mathrm{dual}}}\),
\(\vecb{s}_{\mathrm{guess}}\in\Zq^{n_{\mathrm{guess}}}\), and
\(\vecb{s}_{\mathrm{dual}}\in\Zq^{n_{\mathrm{dual}}}\). Substituting into~\eqref{eq:lwe-instance} gives
\begin{equation}
\label{eq:split-lwe-instance}
\vecb{b}=\mat{A}_{\mathrm{guess}}\vecb{s}_{\mathrm{guess}}+\mat{A}_{\mathrm{dual}}\vecb{s}_{\mathrm{dual}}+\vecb{e} \pmod{q}.
\end{equation}

For any matrix \(\mat{M}\in\Zq^{m\times k}\), we use the associated \(q\)-ary lattice and its dual from Definition~\ref{def:q-ary-lattice}. The dual attack samples a list of short vectors \(\vecb{W}=(\vecb{w}_1,\dots,\vecb{w}_N)\) from the discrete Gaussian \(\mathcal D_{\mathcal{L}_q^\perp(\mat{A}_{\mathrm{dual}}),\,qs}\) on the dual lattice, where \(s\) is the Gaussian width parameter. Their sampling is performed in two stages: a basis of \(\mathcal{L}_q^\perp(\mat{A}_{\mathrm{dual}})\) is BKZ-reduced with block size \(\beta\), then the MCMC algorithm of \cite{wang_lattice_2019} is used to draw \(N\) samples. For each candidate guess \(\widetilde{\vecb{s}}_{\mathrm{guess}}\in\Zq^{n_{\mathrm{guess}}}\), as in the proof of~\cite[Theorem~6]{PoulyShen2024},
\begin{equation}
\vecb{b}-\mat{A}_{\mathrm{guess}}\widetilde{\vecb{s}}_{\mathrm{guess}}
=
\mat{A}_{\mathrm{dual}}\vecb{s}_{\mathrm{dual}}
+
\mat{A}_{\mathrm{guess}}(\vecb{s}_{\mathrm{guess}}-\widetilde{\vecb{s}}_{\mathrm{guess}})
+
\vecb{e}
\pmod q.
\end{equation}

Since \(\vecb{w}_j\in\mathcal L_q^\perp(\mat{A}_{\mathrm{dual}})\), the term involving \(\mat{A}_{\mathrm{dual}}\vecb{s}_{\mathrm{dual}}\) vanishes modulo \(q\) in the inner product
\(\inner{\vecb{w}_j}{\vecb{b}-\mat{A}_{\mathrm{guess}}\widetilde{\vecb{s}}_{\mathrm{guess}}}\).
The resulting cosine score is precisely the empirical score used in~\cite[Algorithm~2]{PoulyShen2024}, the geometric separation condition guaranteeing that the correct guess maximizes this score is therefore inherited from~\cite[Theorem~6]{PoulyShen2024}. We write this empirical score as
\begin{equation}
\label{eq:score-function}
g_{\vecb{W}}(\vecb{y})
=
\frac{1}{N}
\sum_{j=1}^{N}
\cos\left(
\frac{
2\pi
\left\langle
\vecb{w}_j,
\vecb{y}
\right\rangle
}{q}
\right)
\end{equation}
For each candidate guess, the score is
\begin{equation}
F(\widetilde{\vecb{s}}_{\mathrm{guess}})
=
N g_{\vecb{W}}\left(\mat{A}_{\mathrm{guess}}(\vecb{s}_{\mathrm{guess}}-\widetilde{\vecb{s}}_{\mathrm{guess}})+\vecb{e}\right),
\end{equation}
hence maximizing \(F\) is equivalent to maximizing \(g_{\vecb{W}}\).

\begin{lemma}[Lemma~9 of \cite{PoulyShen2024}]
\label{lem:pouly-shen-score-separation}
Let \(\mat{M}\in\Zq^{m\times k}\), let \(\mathcal L\subset\Z^m\) be a lattice, let \(\vecb{e}\in\Z^m\), let \(s,\delta>0\), and let \(N\in\N\). Let \(\tau=\frac{1}{s}\sqrt{m/(2\pi)}\) and \(\eta\geq0\), and assume that \(\lambda_1(\mathcal L+\mathcal L_q(\mat{M}))\geq\tau+\norm{\vecb{e}}\) and
\[
\rho_{1/s}(\vecb{e})
-
\rho_{1/s}\!\left(\lambda_1(\mathcal L+\mathcal L_q(\mat{M}))-\norm{\vecb{e}}-\tau\right)
>
2\delta+\eta.
\]
Then, with probability at least \(1-q^m\cdot2^{-\Omega(N\delta^2)}\) over the choice of \(\vecb{W}=(\vecb{w}_1,\ldots,\vecb{w}_N)\) from \(\mathcal D_{\mathcal L_q^\perp(\mat{M}),qs}^{N}\), we have
\[
g_{\vecb{W}}(\vecb{e})
\geq
\rho_{1/s}(\vecb{e})-\delta
>
\rho_{1/s}\!\left(\lambda_1(\mathcal L+\mathcal L_q(\mat{M}))-\norm{\vecb{e}}-\tau\right)+\delta+\eta
\geq
g_{\vecb{W}}(\vecb{e}+\vecb{x})+\eta
\]
for all \(\vecb{x}\in \mathcal L\setminus\mathcal L_q(\mat{M})\), where \(g_{\vecb{W}}\) is defined in \eqref{eq:score-function}.
\end{lemma}

\begin{algorithm}
\caption{QRS-accelerated modern dual attack}
\label{alg:modern-dual-attack}
\KwIn{\(m\), \(n_{\mathrm{guess}},n_{\mathrm{dual}}\), prime power \(q\), \(N\in\N\), and LWE instance \((\mat{A},\vecb{b})\)}
\KwOut{A guess for \(\vecb{s}_{\mathrm{guess}}\), or \(\bot\)}
\(\vecb{W}_{\mathcal R}=(\vecb{w}_1,\ldots,\vecb{w}_N)\gets N\) samples generated by Algorithm~\ref{alg:qrs-klein-sampler}\;
\(\vecb{s}_{\mathrm{guess}}\gets\bot\)\;
\(F_{\max}\gets0\)\;
\For{\(\widetilde{\vecb{s}}_{\mathrm{guess}}\in\Zq^{n_{\mathrm{guess}}}\)}{
\(y_j\gets\inner{\vecb{w}_j}{\vecb{b}-\mat{A}_{\mathrm{guess}}\widetilde{\vecb{s}}_{\mathrm{guess}}}\pmod q\), for \(j=1,\ldots,N\)\;
\(\displaystyle
F\gets\sum_{j=1}^{N}
\cos\left(
\frac{2\pi y_j}{q}
\right)\)\;
\If{\(F\geq F_{\max}\)}{
\(F_{\max}\gets F\)\;
\(\vecb{s}_{\mathrm{guess}}\gets\widetilde{\vecb{s}}_{\mathrm{guess}}\)\;
}
}
\KwRet{\(\vecb{s}_{\mathrm{guess}}\)}
\end{algorithm}
The analysis of~\cite{PoulyShen2024} introduces a parameter \(\delta>0\) that quantifies how the score separates the correct guess from incorrect ones: under a precondition relating \(\norm{\vecb{e}}\), \(\lambda_1(\mathcal L_q(\mat{A}))\), and \(\delta\) (see~\cite[Theorem~6]{PoulyShen2024} for the exact form), \cite[Algorithm~2]{PoulyShen2024} returns \(\vecb{s}_{\mathrm{guess}}\) with probability at least \(1-q^m\cdot2^{-\Omega(N\delta^2)}\) over the choice of \(\vecb{W}\). The total cost of the attack is
\begin{equation}
\label{eq:dual-attack-cost}
T_{\mathrm{dual}}
=
\operatorname{poly}(m,n)\left(N+q^{n_{\mathrm{guess}}}\right)
+
T_{\mathrm{BKZ}}(m,\beta)
+
N\cdot T_{\mathrm{MCMC}}\!\left(\mathcal L_q^\perp(\mat{A}_{\mathrm{dual}}),qs\right).
\end{equation}
Here \(T_{\mathrm{BKZ}}(m,\beta)\) is the cost of BKZ-\(\beta\) reduction in dimension \(m\), and \(T_{\mathrm{MCMC}}\) is the per-sample MCMC cost with mixing-time dependence as in Theorem~\ref{thm:imhk}. Our method bypasses MCMC by replacing it with QRS as given in Theorem~\ref{thm:qrs}, yielding the following improved complexity result:

\begin{theorem}
\label{thm:qrs-accelerated-dual-attack}
Assume the conditions of \cite[Theorem~6]{PoulyShen2024} hold, where
\(\mat{A}\in \Zq^{m\times n}\), \(\vecb{e}\in\Z^m\),
\(\vecb{s}\in\Zq^n\), \(s,\delta>0\), \(N\in\N\), and \(\varepsilon_{\mathrm{tail}}\in(0,1)\), with
\(\vecb{b}=\mat{A}\vecb{s}+\vecb{e} \bmod q\). Assume the split of
\(\mat{A}\) and \(\vecb{s}\) is as in \eqref{eq:dual-attack-split}. Let
\(\mat{B}\) be a BKZ-\(\beta\)-reduced basis of
\(\mathcal L_q^\perp(\mat{A}_{\mathrm{dual}})\), with Gram--Schmidt vectors
\(\widehat{\vecb{b}}_1,\ldots,\widehat{\vecb{b}}_m\). Set
\begin{equation}
\mathcal{R}\geq
qs\left(
\sqrt{\frac{m}{2\pi}}
+
\sqrt{\frac{\log(N/\varepsilon_{\mathrm{tail}})}{\pi}}
\right).
\end{equation}
Let \(\mathcal{S}_{\mathcal{R}}:=\mathcal L_q^\perp(\mat{A}_{\mathrm{dual}})\cap \mathcal B_m(\mathcal{R})\) and \(\mathcal{X}_{\mathcal{R}}:=\{\vecb{x}\in\Z^m:\mat{B}\vecb{x}\in\mathcal{S}_{\mathcal{R}}\}\). Let \(\vecb{W}_{\mathcal{R}}=(\vecb{w}_1,\ldots,\vecb{w}_N)\) be sampled from \(\mathcal D_{\mathcal{S}_{\mathcal{R}},qs}^{N}\) using
Theorem~\ref{thm:qrs}. Then, Algorithm \ref{alg:modern-dual-attack} on input $(m,n_{\mathrm{guess}},n_{\mathrm{dual}},q,N,(\mat{A},\vecb{b}))$ returns \(\vecb{s}_{\mathrm{guess}}\) with probability at least
\(1-q^m2^{-\Omega(N\delta^2)}-\varepsilon_{\mathrm{tail}}\)
and its end-to-end complexity under the QRS oracle model is
\begin{equation}
\label{eq:qrs-dual-end-to-end-complexity}
\operatorname{poly}(m,n)(N+q^{n_{\mathrm{guess}}})
+
T_{\mathrm{BKZ}}(m,\beta)
+
N\mathcal O\left(\frac{1}{\sqrt{\Delta_{\mathcal{R}}}}\right).
\end{equation}
\end{theorem}

\begin{proof}
Under the dual-attack setup \eqref{eq:dual-attack-split}, let
\(\mathcal{L}_q^\perp:=\mathcal{L}_q^\perp(\mat{A}_{\mathrm{dual}})=\mat{B}\Z^m\). Since $\mathcal{X}_{\mathcal{R}}:=
\{\vecb{x}\in\Z^m:\mat{B}\vecb{x}\in\mathcal{S}_{\mathcal{R}}\},$
applying Theorem~\ref{thm:qrs} to \(\mathcal{X}_{\mathcal{R}}\), with \(\vecb{c}=\vecb{0}\) and Gaussian width \(qs\), gives a coefficient vector \(\vecb{x}\in\mathcal{X}_{\mathcal{R}}\) distributed as
\begin{equation}
\pi_{\mathcal{R}}(\vecb{x})
=
\frac{\rho_{qs}(\mat{B}\vecb{x})}
{\rho_{qs}(\mat{B}\mathcal{X}_{\mathcal{R}})}.
\end{equation}
Since \(\mat{B}\mathcal{X}_{\mathcal{R}}=\mathcal{S}_{\mathcal{R}}\), setting \(\vecb{w}:=\mat{B}\vecb{x}\) gives
\begin{equation}\label{eq:dual-lattice-gaussian}
\mathcal{D}_{\mathcal{S}_{\mathcal{R}},qs}(\vecb{w})
=
\frac{\rho_{qs}(\vecb{w})}
{\rho_{qs}(\mathcal{S}_{\mathcal{R}})},
\end{equation}
where \(\vecb{w}\in\mathcal{S}_{\mathcal{R}}\) is sampled from \eqref{eq:dual-lattice-gaussian}.  The query complexity for one sample is \(\mathcal O(\Delta_{\mathcal{R}}^{-1/2})\), where
\begin{equation}
	\label{eq:Delta_R}
	\Delta_{\mathcal{R}}
=
\frac{\rho_{qs}(\mathcal{S}_{\mathcal{R}})}
{\prod_{i=1}^{m}\rho_{qs/\norm{\widehat{\vecb{b}}_i}}(\Z)}.
\end{equation}
Therefore the list of \(N\) independent samples \(\vecb{W}_{\mathcal{R}}=(\vecb{w}_1,\ldots,\vecb{w}_N)\) from \(\mathcal D_{\mathcal{S}_{\mathcal{R}},qs}^{N}\)
is generated using \(N\cdot\mathcal O\left(\frac{1}{\sqrt{\Delta_{\mathcal{R}}}}\right)\) queries. Including the BKZ cost \(T_{\mathrm{BKZ}}(m,\beta)\), the complexity incurred from sampling is
\begin{equation}
	T_{\mathrm{samp}}=T_{\mathrm{BKZ}}(m,\beta)
+
N\mathcal O\left(\frac{1}{\sqrt{\Delta_{\mathcal{R}}}}\right).
\end{equation}
Define the retained full-Gaussian mass
\begin{equation}
\label{eq:alpha}
\begin{aligned}
\alpha_{\mathcal{R}}
&:=
\sum_{\vecb{w}\in\mathcal{S}_{\mathcal{R}}}
\mathcal D_{\mathcal L_q^\perp,qs}(\vecb{w}) \\
&=
\frac{\rho_{qs}(\mathcal{S}_{\mathcal{R}})}
{\rho_{qs}(\mathcal L_q^\perp)}.
\end{aligned}
\end{equation}
Then for \(\vecb{w}\in \mathcal{S}_{\mathcal{R}}\), one can relate the truncated and the full distribution as
\begin{equation}
\label{eq:truncated-distribution-relation}
\mathcal D_{\mathcal{S}_{\mathcal{R}},qs}(\vecb{w})
=
\frac{\mathcal D_{\mathcal L_q^\perp,qs}(\vecb{w})}{\alpha_{\mathcal{R}}},
\end{equation}
while for \(\vecb{w}\in\mathcal L_q^\perp\setminus \mathcal{S}_{\mathcal{R}}\) we get \(\mathcal D_{\mathcal{S}_{\mathcal{R}},qs}(\vecb{w})=0\). Thus, the total variation distance between \(\mathcal D_{\mathcal L_q^\perp,qs}\) and \(\mathcal D_{\mathcal{S}_{\mathcal{R}},qs}\) can be obtained by splitting the sum into \(\mathcal{S}_{\mathcal{R}}\) and \(\mathcal L_q^\perp\setminus \mathcal{S}_{\mathcal{R}}\):
\begin{align}
d_{\mathrm{TV}}(\mathcal D_{\mathcal L_q^\perp,qs},\mathcal D_{\mathcal{S}_{\mathcal{R}},qs})
&= \frac12 \sum_{\vecb{w}\in \mathcal{S}_{\mathcal{R}}}
\mathcal D_{\mathcal L_q^\perp,qs}(\vecb{w})
\left|1-\frac1{\alpha_{\mathcal{R}}}\right|
+\frac12 \sum_{\vecb{w}\in \mathcal L_q^\perp\setminus \mathcal{S}_{\mathcal{R}}}
\mathcal D_{\mathcal L_q^\perp,qs}(\vecb{w}) \notag\\
&= \frac12 \left(\frac1{\alpha_{\mathcal{R}}}-1\right)
\sum_{\vecb{w}\in \mathcal{S}_{\mathcal{R}}}
\mathcal D_{\mathcal L_q^\perp,qs}(\vecb{w})
+\frac12
\sum_{\vecb{w}\in \mathcal L_q^\perp\setminus \mathcal{S}_{\mathcal{R}}}
\mathcal D_{\mathcal L_q^\perp,qs}(\vecb{w}) \notag\\
&= \frac12 \left(\frac1{\alpha_{\mathcal{R}}}-1\right)
\frac{\rho_{qs}(\mathcal{S}_{\mathcal{R}})}
{\rho_{qs}(\mathcal L_q^\perp)}
+\frac12\left(
1-
\frac{\rho_{qs}(\mathcal{S}_{\mathcal{R}})}
{\rho_{qs}(\mathcal L_q^\perp)}
\right) \notag\\
&= \frac12\left(\frac1{\alpha_{\mathcal{R}}}-1\right)\alpha_{\mathcal{R}}
+\frac12(1-\alpha_{\mathcal{R}}) \notag\\
&= 1-\alpha_{\mathcal{R}} .
\end{align}
For \(N\) independent samples, the union bound gives
\begin{equation}
\label{eq:product-tvd-bound}
d_{\mathrm{TV}}(\mathcal D_{\mathcal L_q^\perp,qs}^{N},\mathcal D_{\mathcal{S}_{\mathcal{R}},qs}^{N})
\leq
N\cdot d_{\mathrm{TV}}(\mathcal D_{\mathcal L_q^\perp,qs},\mathcal D_{\mathcal{S}_{\mathcal{R}},qs})
=
N(1-\alpha_{\mathcal{R}}).
\end{equation}
Given our chosen value of \(\mathcal{R}\), we can bound
\begin{equation}
\label{eq:tail-mass-bound}
1-\alpha_{\mathcal{R}}
=
\frac{\rho_{qs}(\mathcal L_q^\perp\setminus \mathcal{S}_{\mathcal{R}})}
{\rho_{qs}(\mathcal L_q^\perp)}
\leq
\frac{\varepsilon_{\mathrm{tail}}}{N},
\end{equation}
hence \(d_{\mathrm{TV}}(\mathcal D_{\mathcal L_q^\perp,qs}^{N},\mathcal D_{\mathcal{S}_{\mathcal{R}},qs}^{N})
\leq
\varepsilon_{\mathrm{tail}}\). Given that \cite[Algorithm~2]{PoulyShen2024} returns \(\vecb{s}_{\mathrm{guess}}\) with probability at least \(1-q^m2^{-\Omega(N\delta^2)}\), define the event \(\mathcal E
:=
\left\{
\text{Algorithm \ref{alg:modern-dual-attack} returns } \vecb{s}_{\mathrm{guess}}
\right\}\). Since the probability of any event changes by at most the total variation distance between the input distributions, we get
\begin{equation}
\label{eq:success-probability-transfer}
\begin{aligned}
\Pr_{\vecb{W}_{\mathcal{R}}\sim\mathcal D_{\mathcal{S}_{\mathcal{R}},qs}^{N}}[\mathcal E]
&\geq
\Pr_{\vecb{W}\sim\mathcal D_{\mathcal L_q^\perp,qs}^{N}}[\mathcal E]
-
d_{\mathrm{TV}}(\mathcal D_{\mathcal L_q^\perp,qs}^{N},\mathcal D_{\mathcal{S}_{\mathcal{R}},qs}^{N})\\
&\geq
1-q^m2^{-\Omega(N\delta^2)}
-
d_{\mathrm{TV}}(\mathcal D_{\mathcal L_q^\perp,qs}^{N},\mathcal D_{\mathcal{S}_{\mathcal{R}},qs}^{N})\\
&\geq
1-q^m2^{-\Omega(N\delta^2)}
-
\varepsilon_{\mathrm{tail}}.
\end{aligned}
\end{equation}
Finally, since their algorithm contributes a complexity of
\(\operatorname{poly}(m,n)(N+q^{n_{\mathrm{guess}}})\), the total complexity of the dual attack amounts to
\begin{equation}
\label{eq:qrs-dual-total-complexity}
\operatorname{poly}(m,n)(N+q^{n_{\mathrm{guess}}})
+
T_{\mathrm{BKZ}}(m,\beta)
+
N\mathcal O\left(\frac{1}{\sqrt{\Delta_{\mathcal{R}}}}\right).
\end{equation}
\qed
\end{proof}

Following \cite{PoulyShen2024}, \(\Delta\) is defined as
\begin{equation}\label{eq:delta-PoulyShen2024}
\Delta :=
\frac{\rho_{qs}(\mathcal L_q^\perp(\mat{A}_{\mathrm{dual}}))}
{\prod_{i=1}^{m}\rho_{qs/\norm{\widehat{\vecb{b}}_i}}(\Z)}.
\end{equation}
By the definitions of \(\alpha_{\mathcal{R}}\) and \(\Delta_{\mathcal{R}}\) in \eqref{eq:alpha} and \eqref{eq:Delta_R}, respectively, we have the identity
\begin{equation}
\Delta_{\mathcal{R}}=\alpha_{\mathcal{R}}\Delta.
\end{equation}
Therefore the Gaussian mass outside of the truncation ball can be bounded above as
\begin{equation}
1-\alpha_{\mathcal{R}} \leq \frac{\varepsilon_{\mathrm{tail}}}{N}
\implies
\alpha_{\mathcal{R}} \geq 1-\frac{\varepsilon_{\mathrm{tail}}}{N},
\end{equation}
Hence
\begin{equation}
\frac{1}{\sqrt{\Delta_{\mathcal{R}}}}
\leq
\frac{1}{
\sqrt{\displaystyle
\left(1-\frac{\varepsilon_{\mathrm{tail}}}{N}\right)\Delta
}
}.
\end{equation}

Given that we choose the truncation radius according to Corollary~\ref{cor:tail-bound}, most of the Gaussian mass is covered. Since the Gaussian mass decays exponentially fast outside the chosen radius, we get \(\alpha_{\mathcal{R}} \rightarrow 1\), hence \(\displaystyle \frac{1}{\sqrt{\Delta_{\mathcal{R}}}}
\rightarrow
\frac{1}{\sqrt{\Delta}}\). Finally, under the above conditions, this gives a quadratic speedup in the sampling term of the dual attack framework:

\begin{equation}
\label{eq:dual-total-complexity}
T_{\mathrm{dual}}=\operatorname{poly}(m,n)(N+q^{n_{\mathrm{guess}}})
+
T_{\mathrm{BKZ}}(m,\beta)
+
N\mathcal O\left(\frac{1}{\sqrt{\Delta}}\right).
\end{equation}

We use the reduced complexity in \eqref{eq:qrs-dual-total-complexity} to estimate fixed-parameter costs for Kyber. \Cref{tab:qrs-vs-pouly-shen-no-ms} gives the no-modulus-switching comparison, using the same parameter choices \((m,n_{\mathrm{guess}},n_{\mathrm{dual}},\beta,s)\) as in~\cite{PoulyShen2024} and replacing only the Gaussian-sampling term. The estimated costs decrease from \(185\) to \(176\), from \(273\) to \(269\), and from \(376\) to \(363\) bits for Kyber512, Kyber768, and Kyber1024, respectively. Since \Cref{thm:qrs-accelerated-dual-attack} is stated for the no-modulus-switching Pouly--Shen framework, the modulus-switching figures in \Cref{tab:qrs-vs-pouly-shen-ms} should be read only as an indicative fixed-parameter comparison. In this case, the sampling improvement is largely absorbed by the remaining cost terms for Kyber512 and Kyber768, while Kyber1024 still drops from \(279\) to \(261\) bits.

\begin{table}[!htbp]
\centering
\begin{tabular}{c|cc|ccccc}
\multicolumn{8}{c}{\textbf{Without Modulus Switching}}\\[0.25em]
\hline
&
\multicolumn{2}{c|}{Attack Cost}
&
\multicolumn{5}{c}{Parameters}
\\
Scheme
& \cite{PoulyShen2024}
& This Work
& \(m\)
& \(n_{\mathrm{guess}}\)
& \(n_{\mathrm{dual}}\)
& \(\beta\)
& \(s\)
\\
\hline
Kyber512  & 185 & 176 & 1013 & 15 & 497 & 550  & 0.200\\
Kyber768  & 273 & 269 & 1469 & 23 & 745 & 870  & 0.260\\
Kyber1024 & 376 & 363 & 2025 & 31 & 993 & 1230 & 0.270\\
\hline
\end{tabular}
\vspace{0.6\baselineskip}
\caption{Comparison with the dual attack costs without modulus switching reported in~\cite{PoulyShen2024}. All costs are base-2 logarithms. The ``This Work'' column keeps the same parameters \((m,n_{\mathrm{guess}},n_{\mathrm{dual}},\beta,s)\) and substitutes only the QRS Gaussian-sampling term.}
\label{tab:qrs-vs-pouly-shen-no-ms}
\end{table}

Separately, since \Cref{thm:qrs} prepares a quantum state encoding a dual-lattice Gaussian, our quantum sampler can be combined with the quantum dual attack framework of Albrecht and Shen~\cite{AlbrechtShen2022}. In this composition, our QRS state-preparation oracle would replace the precomputed Gaussian sample list and then be used in the quantum mean estimation procedure of~\cite[Theorem~5]{AlbrechtShen2022} to estimate the score for each guess, enabling a fully quantum end-to-end dual attack that is QRACM-free. A more detailed treatment of this composition appears in the recent concurrent work of~\cite{cryptoeprint:2026/984}, to which we refer the reader for the full analysis.

\begin{table}[!htbp]
\centering
\begin{tabular}{c|cc|ccccc}
\multicolumn{8}{c}{\textbf{With Modulus Switching}}\\[0.25em]
\hline
&
\multicolumn{2}{c|}{Attack Cost}
&
\multicolumn{5}{c}{Parameters}
\\
Scheme
& \cite{PoulyShen2024}
& This Work
& \(m\)
& \(n_{\mathrm{guess}}\)
& \(n_{\mathrm{dual}}\)
& \(\beta\)
& \(s\)
\\
\hline
Kyber512  & 141 & 141 & 763  & 141 & 371 & 390 & 0.170\\
Kyber768  & 202 & 201 & 1169 & 201 & 567 & 610 & 0.240\\
Kyber1024 & 279 & 261 & 1575 & 261 & 763 & 890 & 0.260\\
\hline
\end{tabular}
\vspace{0.6\baselineskip}
\caption{Comparison with the dual attack costs with modulus switching reported in~\cite{PoulyShen2024}. All costs are base-2 logarithms. The ``This Work'' column keeps the same parameters \((m,n_{\mathrm{guess}},n_{\mathrm{dual}},\beta,s)\) and substitutes only the QRS Gaussian-sampling term.}
\label{tab:qrs-vs-pouly-shen-ms}
\end{table}

\section{Quantum Speedup of Trapdoor Sampling}
\label{sec:falcon-trapdoor-sampling}

We now describe how the same QRS idea can be used as
a replacement for the Gaussian sampling subroutine in GPV/FALCON-type
trapdoor sampling. This section is not meant to give a new implementation of
FALCON's signing algorithm, but rather it identifies the precise sampling step in
which the MCMC sampler in~\cite{wang_lattice_2019} can be
replaced by our QRS sampler.

\subsection{GPV Trapdoor Sampling}

In~\cite{GenPeiVai2008}, key generation produces two
bases of the same lattice: a hard public basis and a short private basis. The
public basis is used for verification, while the short private basis is the trapdoor used for signing. Given a message, the signer obtains a coset representative by hashing, denoted here by \(\vecb c\), and samples a short vector from a discrete Gaussian over the corresponding coset \(\mathcal D_{\mathcal L+\vecb c,s}\).
The signature is the resulting short lattice-coset vector, and the verifier checks both the coset relation and the norm bound using the public key. The purpose of the trapdoor sampler is therefore to output a vector whose distribution is close to the target lattice Gaussian while hiding the private
basis used to sample it.

FALCON~\cite{FouqueEtAlFalcon2020} instantiates this GPV paradigm over NTRU lattices. Following the notation in~\cite{wang_lattice_2019}, let \(R=\mathbb Z[x]/(x^d+1)\), where \(d\) is a power of two, and let the secret
polynomials \(f,g\in R\) satisfy that \(f\) is invertible. One finds
\(F,G\in R\) such that
\[
    fG-gF = q \bmod (x^d+1).
\]
The private NTRU basis is
\[
    \mat{B} =
    \begin{pmatrix}
        C(g) & -C(f)\\
        C(G) & -C(F)
    \end{pmatrix}^{T},
\]
where \(C(\cdot)\) denotes the negacyclic matrix associated with a polynomial. The corresponding public basis is
\[
    \mat{A} =\begin{pmatrix}
        -C(h) & I_d\\
        qI_d  & \mathbf{0}_d
    \end{pmatrix}^{T},
    \qquad h = g/f \bmod q.
\]
Both bases generate the same NTRU lattice
\(
    \mathcal{L} = \{(u,v)\in R^2 : u+vh=0 \bmod q\},
\)
whose dimension is \(n=2d\). In the signing step, the short basis \(\mat{B}\) is used
to sample from a coset Gaussian distribution associated with the message digest.

\subsection{Quantum-Accelerated GPV Sampling}
\label{subsec:quantum-accelerated-gpv}

We replace only the trapdoor Gaussian sampler in the GPV/FALCON signing
pipeline. The key generation, hash-to-coset map, public verification, and
trapdoor basis are unchanged; only the sampler for
\(\mathcal D_{\mathcal L+\vecb c,s}\) is replaced. Wang and Ling instantiate this step
with the independent MHK sampler~\cite{wang_lattice_2019}. Here we use the
truncated Klein proposal as the coherent proposal oracle for QRS.

Let \(\widehat{\vecb{b}}_1,\ldots,\widehat{\vecb{b}}_n\) be the Gram--Schmidt
vectors of the trapdoor basis \(\mat{B}\) and set $s_i=s/\|\widehat{\vecb{b}}_i\|$,
\(B_{\max}=\max_i\|\widehat{\vecb{b}}_i\|\). By Theorem~\ref{thm:imhk},
for the independent MHK sampler, the convergence
rate is controlled by a  \(\Delta\). The classical MCMC sampling cost is
given by
\begin{equation}\label{eq:Delta2}
T_{\mathrm{MCMC}}\asymp \frac{1}{\Delta}=\frac{\prod^n_{i=1}\rho_{s_i}(\mathbb{Z})}{\rho_{s,\vecb c}(\mathcal L)}.
\end{equation} 
This expression depends on the full Gram--Schmidt profile of the trapdoor
basis and on the coset center \(\vecb c\). Wang and Ling further derive a
basis-profile-independent upper bound above smoothing~\cite[Eq. 77]{wang_lattice_2019}: if
\(s\geq \sqrt{\gamma}B_{\max}\) for some \(\gamma\geq1\), then
\begin{equation}
    T_{\mathrm{MCMC}}
    \lesssim     \vartheta_3(\gamma)^n(1+2\varepsilon),\,\,\vartheta_3(\tau)=\sum_{k\in\mathbb Z}e^{-\pi\tau k^2}.
    \label{eq:falcon-wang-gamma-bound}
\end{equation}
In the numerical comparison below we use this latter upper bound, rather than \eqref{eq:Delta2}.

By Theorem~\ref{thm:qrs}, after truncating the coefficient space so that the
tail mass is negligible, QRS prepares the target coset Gaussian using
\(\mathcal O(1/\sqrt{\Delta_{\mathcal R}})\) oracle queries. Since
\(\Delta_{\mathcal R}\approx\Delta\) for such a truncation, the same bound gives
\begin{equation}
    T_{\mathrm{QRS}}
    \lesssim \vartheta_3(\gamma)^{n/2}\sqrt{1+2\varepsilon}.
    \label{eq:falcon-qrs-gamma-bound}
\end{equation}
Thus the QRS replacement turns the MCMC sampling cost into its square
root.

\begin{figure}[t]
    \centering
    \includegraphics[width=0.78\linewidth]{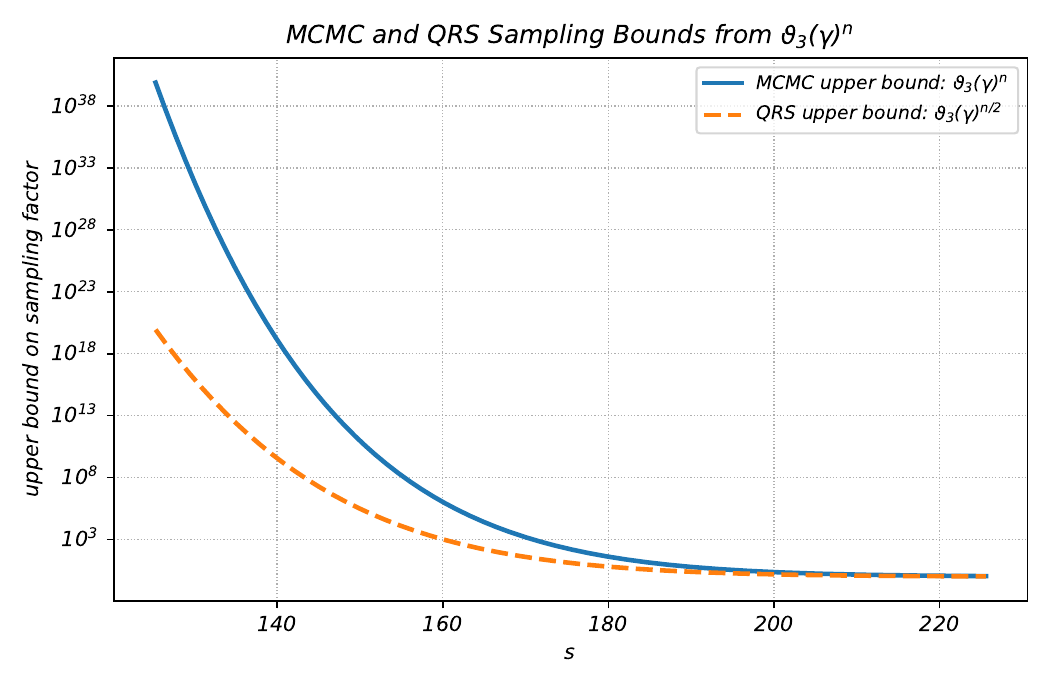}
    \caption{MCMC and QRS upper bounds derived from~\cite{wang_lattice_2019}'s
    \(\vartheta_3(\gamma)^n\) estimate for FALCON-512. The horizontal axis
    uses the Gaussian parameter \(s\), with \(\gamma=(s/B_{\max})^2\).}
    \label{fig:falcon-qrs-theta-bound}
\end{figure}

For FALCON-512, we take parameters \(d=512\), \(q=12289\), lattice dimension
\(n=2d=1024\), and report \(B_{\max}=127\) for one generated NTRU trapdoor
basis~\cite[Example~1]{wang_lattice_2019}. Taking \(\varepsilon\approx0\), the
corresponding profile-free upper bounds are shown in
\Cref{tab:falcon-gamma-bound-qrs}.

\begin{table}[t]
\centering
\small
\begin{tabular}{ccccc}
\toprule
\(\gamma\) & \(s/B_{\max}\) & \(\sigma/B_{\max}\) &
\(\vartheta_3(\gamma)^n\) &
\(\vartheta_3(\gamma)^{n/2}\) \\
\midrule
\(\pi/2\) & 1.253 & 0.500 & \(2.24\times10^{6}\) & \(1.50\times10^{3}\) \\
2         & 1.414 & 0.564 & 45.49                  & 6.74 \\
3         & 1.732 & 0.691 & 1.18                   & 1.09 \\
\bottomrule
\end{tabular}
\caption{Upper-bound evaluation of the MCMC and QRS sampling factors for
FALCON-512. The fourth column upper-bounds \(T_{\mathrm{MCMC}}\), while the
fifth column upper-bounds \(T_{\mathrm{QRS}}\). We use \(n=1024\),
\(s=\sqrt{2\pi}\sigma\), \(B_{\max}=127\), and \(\varepsilon\approx0\).}
\label{tab:falcon-gamma-bound-qrs}
\end{table}

\Cref{fig:falcon-qrs-theta-bound} and \Cref{tab:falcon-gamma-bound-qrs} should
be interpreted as worst-case upper-bound comparisons depending only on
\(B_{\max}\). They do not reproduce the instance-dependent curve in~\cite[Fig.~4]{wang_lattice_2019}, which uses the full Gram--Schmidt profile~\eqref{eq:Delta2} of a sampled
FALCON key and can be much smaller than the profile-free bound in
\eqref{eq:falcon-wang-gamma-bound}.

\section{Conclusion}
\label{sec:conclusion}

In this work, we showed that the lower-bound relation behind Klein's proposal
distribution can be used as the domination condition required for quantum
rejection sampling. After truncating the coefficient space, QRS prepares the
truncated lattice Gaussian with query complexity
\[
    \mathcal O\!\left(\frac{1}{\sqrt{\Delta_{\mathcal R}}}\right),
\]
while the truncation error can be made negligible by choosing a sufficiently
large radius. We applied this sampler to the modern dual attack framework of
Pouly and Shen~\cite{PoulyShen2024}. Replacing the classical IMHK sampler by QRS changes the Gaussian-sampling contribution from a \(1/\Delta\) dependence to a \(1/\sqrt{\Delta_{\mathcal R}}\) dependence, yielding lower fixed-parameter attack-cost estimates for Kyber. We also showed that the same idea applies to GPV/FALCON trapdoor sampling using the upper bound \(\vartheta_3(\gamma)^n\) from~\cite{wang_lattice_2019}, for which the corresponding QRS bound becomes \(\vartheta_3(\gamma)^{n/2}\).

\bibliographystyle{splncs04}
\bibliography{refs}

\appendix
\renewcommand{\theHsection}{appendix.\Alph{section}}

\end{document}